\documentclass[journal,10pt]{IEEEtran}
\usepackage{booktabs} 
\usepackage{amsopn}
\usepackage{graphicx}
\usepackage{mathtools}
\usepackage{amsmath}
\usepackage{multirow}
\usepackage[ruled]{algorithm2e}
\usepackage[noend]{algpseudocode}
\usepackage{tabularx}
\usepackage{epstopdf}
\usepackage{amsthm}
\usepackage{mathtools}
\usepackage{pgfplots}
\usepackage{caption}
\usepackage{amsmath, pgfplots, graphicx, siunitx}
\usepackage{amsmath,pgfplots,gincltex}
\usepackage[version=3]{mhchem}

\newtheorem{prop}{Proposition}

\begin{document}

\title{Topology Control for Energy-Efficient Localization in Mobile Underwater Sensor Networks using Stackelberg Game}

\author{Yali~Yuan,~
        Chencheng~Liang,~\IEEEmembership{Member,~IEEE,}
        Megumi~Kaneko,~\IEEEmembership{Senior Member,~IEEE,}\\
        Xu~Chen,~\IEEEmembership{Member,~IEEE,}
        and~Dieter~Hogrefe
\thanks{Yali Yuan, Chencheng Liang and Dieter Hogrefe are in Telematics Group, Institute of Computer Science, University of Goettingen, 37077 Goettingen, Germany.}
\thanks{Megumi Kaneko is at Information Systems Architecture Science Research Division, National Institute of Informatics, 101-8430 Tokyo, Japan}
\thanks{Xu Chen is in School of Data and Computer Science, Waihuandong Road, Sun Yat-sen University, Guangzhou, China}
\thanks{Co-first author: Chencheng Liang}
\thanks{Corresponding author: Yali Yuan (yali.yuan@informatik.uni-goettingen.de)}
}

\maketitle

\begin{abstract}
The characteristics of mobile Underwater Sensor Networks (UWSNs), such as low communication bandwidth, large propagation delay, and sparse deployment, pose challenging issues for successful localization of sensor nodes. In addition, sensor nodes in UWSNs are usually powered by batteries whose replacements introduce high cost and complexity. Thus, the critical problem in UWSNs is to enable each sensor node to find enough anchor nodes in order to localize itself, with minimum energy costs. In this paper, an Energy-Efficient Localization Algorithm (EELA) is proposed to analyze the decentralized interactions among sensor nodes and anchor nodes. A Single-Leader-Multi-Follower Stackelberg game is utilized to formulate the topology control problem of sensor nodes and anchor nodes by exploiting their available communication opportunities. In this game, the sensor node acts as a leader taking into account factors such as `two-hop' anchor nodes and energy consumption, while anchor nodes act as multiple followers, considering their ability to localize sensor nodes and their energy consumption. We prove that both players select best responses and reach a socially optimal Stackelberg Nash Equilibrium. Simulation results demonstrate that the proposed EELA improves the performance of localization in UWSNs significantly, and in particular the energy cost of sensor nodes. Compared to the baseline schemes, the energy consumption per node is about $48\%$ lower in EELA, while providing a desirable localization coverage, under reasonable error and delay.
\end{abstract}

\begin{IEEEkeywords}
Underwater sensor networks, localization, energy consumption, topology control, Stackelberg game
\end{IEEEkeywords}

\IEEEpeerreviewmaketitle
\section{Introduction}
Localization of mobile sensor nodes is indispensable for enabling Underwater Wireless Sensor Networks (UWSNs). The gathered data is not useful if it is not correlated to a specific position of the sensor node. Many applications, such as aquatic environment monitoring, target tracking \cite{bochem2016tri}, geo-routing protocols~\cite{xie2006vbf} and pollution control, require the location information. However, Global Positioning Systems (GPS) cannot be used in UWSNs because of the high attenuation of radio signal and their power hungry nature. Moreover, UWSNs are mostly based on acoustic communication systems which suffer diverse issues stemming from the aquatic conditions, such as frequency dispersion, multipath fading, limited bandwidth and energy \cite{akyildiz2005underwater, tan2011survey}.

One important challenge is that underwater sensor nodes have limited resources due to the use of non-rechargeable batteries, which directly determines the network life time. Given the engineering hurdles and financial costs of battery replacement, the design of energy-efficient localization techniques becomes critical to extend the network lifetime in UWSNs. The sensor nodes' energy is consumed mainly by packet transmission and reception, which is much larger than that of the idle listening \cite{freitag2005whoi}, so adjusting the transmission power by topology control is one possible way to save energy in UWSNs. In most localization systems of UWSNs \cite{tan2011survey}, multiple anchor nodes are required to help one sensor node to find its position. The performance of these localization methods depend on different factors such as the nodes' initial reference position, number of sensor nodes, number of anchor nodes, ranging technique as well as the position of anchor nodes.

These facts motivate us to seek for energy-efficient solutions enabling each sensor node to find the required number of anchor nodes in view of localization by topology control. Topology control in this paper represents the process of controlling the amount/quality of neighboring nodes by transmission power control, where anchor nodes assist sensor nodes to find their locations in UWSNs. So far, only a limited number of schemes have been proposed for the localization service in UWSNs \cite{teymorian20093d}, \cite{zhou2011scalable}, \cite{luo2016localization}, especially for localization in sparse UWSNs by using topology control \cite{misra2015game}. However, the proposed scheme in \cite{misra2015game} only considers energy-saving for anchor nodes, although the sensor nodes deployed underwater consume the bulk of the energy, causing the network lifetime to decrease prematurely.

In this paper, we leverage the benefits of topology control to achieve a high localization performance and a low energy consumption in UWSNs. In such systems, in order to successfully localize a single sensor node, multiple anchor nodes are required. This motivates us to model this problem as a Single-Leader-Multi-Follower Stackelberg game, for which we define new utility metrics for trading-off localization ability and energy efficiency with no need for new equipment, nor extra cost, unlike the existing state-of-the-art schemes in \cite{misra2015game} and \cite{isik2009three}. In summary, the contributions of this paper are given as follows,
\begin{enumerate}
\item [1)] The considered localization problem is analyzed as a Single-Leader-Multi-Follower Stackelberg game whereby the trade-off between localization ability and energy consumption are considered in the utility functions. Optimal transmit powers for sensor nodes and anchor nodes are derived, and are shown to achieve Nash Equilibrium.
\item [2)] Based on our analysis, we propose the EELA (Energy-Efficient Localization Algorithm) that builds the strategic interactions among each sensor node and multiple anchor nodes for enabling energy-efficient localization.
\item [3)] Extensive numerical evaluations show that the proposed EELA scheme achieves about $48\%$ energy reduction for all nodes on average compared to the state-of-the-art approach, while achieving high quality localization coverage under reasonable error and delay levels.
\end{enumerate}

The rest of the paper is organized as follows. Section \ref{sec:relatedWork} discusses the related works. Section \ref{se:systemModel} introduces the system model. The detailed description of the proposed EELA scheme is presented in Section \ref{sec:ProblemSolu}. Simulation results and performance evaluation are shown in Section \ref{sec:Simulations}. Finally, Section \ref{sec:Conclusion} presents the conclusion and future work.

\section{Related Work}
\label{sec:relatedWork}

\subsection{Localization in UWSNs}
A set of localization techniques has been proposed for UWSNs in recent years. A detailed survey about these works was given in \cite{han2012localization} and \cite{tuna2017survey}. One localization scheme was presented in \cite{syed2006time}, which worked well in high latency networks and improved the energy efficiency as well. One-way and two-way MAC-layer message delivery were combined in this paper. However, it is only suitable for static UWSNs due to its assumption of constant propagation delays among sensor nodes. In \cite{isik2009three}, three-Dimensional Underwater Localization (3DUL) described a distributed, iterative and dynamic solution to the localization problem in the underwater acoustic sensor network with only three anchor nodes at the surface of the water. Trilateration algorithm was used to estimate the sensor node location. However, the error accumulated with the iteration increase, leading to inaccurate positioning of later sensor nodes that require the location information of new anchor nodes generated by sensor nodes in 3DUL. One localization scheme \cite{zhou2011scalable} was proposed in a hierarchical underwater sensor networks consisting of surface buoys, anchor nodes, and ordinary nodes. Sensor nodes were localized by the trilateration method. However, the node density is assumed to be high in this scheme due to the long distance acoustic communication between the anchor node and the surface buoy. In \cite{luo2016localization}, a novel scheme was proposed for long-term maritime surveillance monitoring tasks in ocean sensor networks. Liu et al. \cite{liu2016joint} proposed a joint solution for localization and time synchronization in mobile underwater sensor networks. The stratification effect of underwater medium was taken into account in localization. Schemes utilizing Autonomous Underwater Vehicles (AUVs) \cite{erol2007auv}, \cite{waldmeyer2011multi} and \cite{ojha2013hasl} as beacon nodes result in additional cost to the network.
\subsection{Topology Control}
The problem of topology control for WSNs has been extensively studied in recent years. The topology control scheme presented in \cite{li2005localized} and \cite{sethu2010new} started with neighbor finding, where all nodes transmit at their maximum transmission power. Later, each node computes the minimum transmission power required to maintain network connectivity. A game theoretic model of topology control to analyze the decentralized interactions among heterogeneous sensors was given in \cite{ren2009game}. The connectivity, the success rate and the power consumption were considered to achieve desirable network performances. Zhu et al. \cite{zhu2017game} took into account the signal to interference plus noise ratio and power efficiency to solve the distributed power control issues in cognitive wireless sensor network with imperfect information. A game-theoretic power control mechanism based on the Hidden Markov Model (HMM) was employed to maximize the network lifetime.

Although there are increasing interests in UWSNs in the past several years, few works have been investigated on approaches for topology control in UWSNs, especially for the localization by using topology control. In \cite{liu2012topology}, a distributed radius determination algorithm is designed for the mobility-based topology control problem. However, the energy consumption of message reception was not considered in this paper. A scale-free network model for calculating edge probability was employed to generate the initial topology randomly by Liu et al. \cite{liu2014complex}. In order to ensure the connectivity and coverage of the network, two kinds of cluster-heads were constructed by using a topology control strategy based on complex network theory. A Single-Leader-Multi-Follower Stackelberg game, called Opportunistic Localization by Topology Control (OLTC) scheme, was proposed in \cite{misra2015game} to build a localization model with high coverage and less energy consumption in sparse UWSNs. Trilateration algorithm was employed to localize the sensor node. However, in this paper, each sensor node always uses the maximum transmission power to broadcast the `Request' message, which results in more energy consumption. Besides, only the energy-saving for anchor nodes is considered in OLTC. However, in many scenarios of UWSNs, energy-saving for sensor nodes is much more important, since they have a limited battery and are deployed underwater, hence directly affecting the lifetime of the whole network.
\section{System Model}
\label{se:systemModel}
\subsection{System Overview}
The proposed EELA is implemented in the 3D UWSN, where $\left\{N_{a}\right\}$ is the set of anchor nodes deployed on the surface of water and $\left\{N_{s}\right\}$ denotes the set of sensor nodes deployed underwater. All nodes move passively given the water wave and underwater currents. Fig. \ref{fi:struScena} depicts the deployment scenario of EELA, where the cylinders represent anchor nodes on the water surface while dotted circles express sensor nodes which are randomly positioned underwater. Each sensor node $i$ has multiple neighboring anchor nodes within the current transmission range under power $P_{i}$, which is expressed by $n_{neig}(P_{i})$. Both sensor nodes and anchor nodes can change their transmission power to maximize their own benefits. The transmission power value for each sensor node is within $[0, P_{max}]$, corresponding to different transmission ranges within $[0,\ R_{max}]$ (see Eq. \eqref{eq:tp}). The transmission power value for each anchor node is within $[0, Q_{max}]$, corresponding to different transmission ranges within $[0,\ R_{max}]$ (see Eq. \eqref{eq:tp}). We set $P_{max}=Q_{max}$.
\begin{figure}[!t]
\centering
\includegraphics[width=0.45\textwidth]{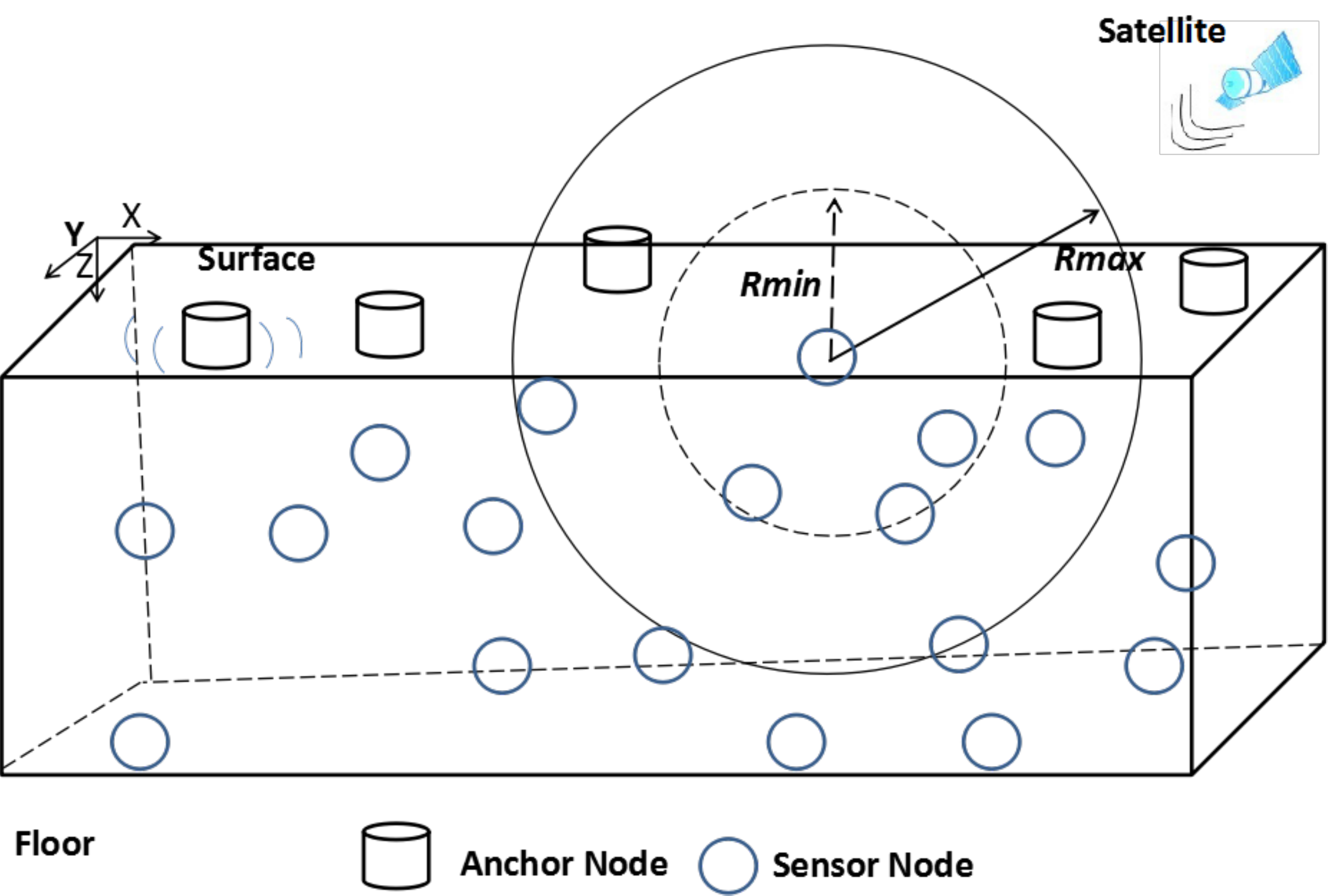}
\caption{Deployment scenario of the proposed EELA scheme}
\label{fi:struScena}
\end{figure}
For reference purposes, a list of symbols used in the description of our scheme is given in Table~\ref{ta:notation}.
We have the following assumptions as in \cite{misra2015game} in the design of EELA.
\begin{enumerate}
\item [1)] All nodes are time-synchronized.
\item [2)] Sensor nodes are randomly deployed underwater while anchor nodes randomly float on the water surface.
\item [3)] Sensor nodes are aware of their depth.
\end{enumerate}
\begin{table}[t]
\caption{Symbols used in the proposed EELA scheme}
\centering
\begin{tabularx}{0.49\textwidth}{c  X }
\hline
Parameter & Description \\ \hline
    $\left\{N_{s}\right\}$ & The set of sensor nodes \\
    $\left\{N_{a}\right\}$  & The set of anchor nodes \\
    $\{Nbr_{i}^{s}\}$ & The set of neighboring sensor nodes of $i$th sensor node\\
    $\{Nbr_{j}^{a}\}$ & The set of neighboring anchor nodes of $j$th anchor node\\
    $n_{h}^{req}$ & Additional number of anchor nodes required by $h$th sensor node\\
    $E_{tl}$ & Total remaining energy per node\\
    $C_{j}$ & Transmission energy cost per unit power of $j$th anchor node\\
    $E_{i}$ & Transmission energy cost per unit power of $i$th sensor node\\
    $P_{i}$ & Transmission power of $i$th sensor node \\
    $Q_{j}$ & Transmission power of $j$th anchor node\\
    $P_{max}$ & Maximum transmission range of sensor node\\
    $Q_{max}$ & Maximum transmission range of anchor node\\
    ${OA}_{j}$ & Localization ability of $j$th anchor node\\
    ${OS}_{i}$ & Ability of sensor node $i$ to find anchor nodes\\
\hline
\end{tabularx}
\label{ta:notation}
\end{table}
\subsection{Propagation Model}
\label{subse:proModel}
According to the underwater propagation model \cite{harris2007modeling}, the transmission power required by a sending node to a receiving node is given in Eq. \eqref{eq:tp}, where $P_{0}$ is the received signal strength,
\begin{equation}
\label{eq:tp}
    P = A(R, f) + P_{0}.
\end{equation}
In general, acoustic communications are used in underwater environment due to the small attenuation of acoustic signals \cite{erol2011survey}. The attenuation $A(R, f)$ in an underwater acoustic channel for a signal with frequency $f$ over a distance $R$ is given as
\begin{equation}
    \label{eq:attenuation}
     A(R, f) = A_{norm}R^{k}a(f)^{R},
 \end{equation}
where $A_{norm}$ is a normalization constant; $R$ is the distance in meters between the sender and the receiver; $f$ is the frequency; $a(f)$ is the absorption coefficient in dB/m. The spreading factor $k$ is an expression of the geometry of propagation where typically $k=1.5$ \cite{erol2011survey}.
Eq. \eqref{eq:highFre} describes the absorption coefficient with values in dB/km \cite{stojanovic2007relationship},
\begin{equation}
\label{eq:highFre}
10\log a(f)=\left\{
\begin{array}{ll}
0.003 + 0.11\frac{f^{2}}{1+f^{2}} \quad &\multirow{2}*{$f\geq 0.4$}\\
+ 44\frac{f^{2}}{4100+f^{2}} + 2.75.10^{-4}f^{2}&\\
\specialrule{0em}{0.5ex}{0.5ex}
0.002 + 0.11\frac{f^{2}}{1+f^{2}} \quad &\multirow{2}*{$f < 0.4$}\\
 + 0.011f^{2}&\\
\end{array}\right.
\end{equation}
where $T$ is the Thorp`s approximation for absorption loss in dB/km, and $f$ is center frequency in kHz.

Different communication ranges correspond to different bandwidths. For example, if the distance range is 10km to 100km, the bandwidth is limited to few kHz. However, 10 kHz matches the short range (from 1km to 10km) and a few hundred kHz bandwidth is available for ranges below 100m \cite{stojanovic1995underwater}.
\section{Problem Formulation and Solution}
\label{sec:ProblemSolu}
\subsection{Proposed Single Leader Multi Follower Stackelberg Game Formulation}
In Stackelberg game \cite{sherali1983stackelberg}, two types of players (leader and followers) are used to model the hierarchy of actions. The leader moves first and selects a strategy. Based on the action of the leader, the followers choose best response strategies that maximize their utilities. Then, the leader selects one strategy to maximize its utility based on the strategies of followers. In a distributed localization scenario, a sensor node can localize itself after receiving enough location beacons from multiple anchor nodes. However, due to the random and sparse node deployment, and the mobility of nodes, a sensor node may not find enough neighbor anchor nodes. A Single-Leader-Multi-Follower Stackelberg game \cite{sherali1983stackelberg} is employed, where the sensor node acts as the single leader while anchor nodes are multiple followers. The sensor node acts first and chooses a transmission power to send a request message, which is similar to the leader releasing a price in Stackelberg game. Each anchor node reacts, i.e, it selects a transmission power to send reply message, after the action of the sensor node.
\subsection{Utility Function of Anchor Nodes}
Anchor nodes are followers. They decide their strategies to handle the maximum number of requests from sensor nodes with minimum energy consumption.

Let $n_{i}^{req}$ be the additional number of anchor nodes required by sensor node $i$ for localization, which is defined by,
\begin{equation}
\label{eq:reqi}
n_{i}^{req}=\begin{cases}
n_{min}^{req} - |\mathcal{V}_{i}|, & \text{if $|\mathcal{V}_{i}| < n_{min}^{req}$}\\
0, & \text{otherwise,}
\end{cases}
\end{equation}
where $n_{min}^{req}$ represents the number of anchor nodes required for one sensor node to get its location and $\mathcal{V}_{i}$ is the set of the anchor nodes in the communication range of sensor node $i$.

The localization ability $OA_{j}(Q_{j}, P_{1}, P_{2},\cdots, P_{n_{arx}})$ of anchor node $j$ is composed of several terms expressing different effects,
\begin{align}
\label{eq:obability}
OA_{j}&(Q_{j}, P_{1}, P_{2},\cdots, P_{n_{arx}})= \quad \nonumber \\&\frac{n_{hd}(Q_{j})}{n_{arx}}+\frac{n_{hd}(Q_{j})}{\sum_{k=1}^{n_{arx}}n_{k}^{req}}
-\frac{\sum_{i=1}^{n_{arx}}P_{i}}{Q_{j}}.
\end{align}
In \eqref{eq:obability}, the first two terms $\frac{n_{hd}(Q_{j})}{n_{arx}}$ and $\frac{n_{hd}(Q_{j})}{\sum_{k=1}^{n_{arx}}n_{k}^{req}}$ are similar to those in the utility function in \cite{misra2015game}. $\frac{n_{hd}(Q_{j})}{n_{arx}}$ is the `ability of $j$th anchor node to resolve sensor requests', where $n_{hd}(Q_{j})$ is the number of requests that can be handled by anchor node $j$ with transmission power $Q_{j}$ and $n_{arx}$ is the total number of request messages received from sensor nodes. From Proposition~\ref{eq:sensorNneig} below, we can see that $n_{hd}(Q_{j})$ of a follower (anchor node) $j$ is non-decreasing with the increase of the transmission power $Q_{j}$. The second term $\frac{n_{hd}(Q_{j})}{\sum_{k=1}^{n_{arx}}n_{k}^{req}}$ is the `ability of $j$th anchor node to serve additional demands'. It means that only $n_{hd}(Q_{j})$ requests can be served among the total sum demand for additional anchor nodes $\sum_{k=1}^{n_{arx}}n_{k}^{req}$ from sensor nodes. Finally, the third term $\frac{\sum_{i=1}^{n_{arx}}P_{i}}{Q_{j}}$ expresses the relation between the sum-transmit power received from sensor nodes and the transmission power of the anchor node $j$. If the transmission power $\sum_{i=1}^{n_{arx}}P_{i}$ received from sensor nodes increases, anchor node $j$ has to handle more sensor nodes. Therefore, the localization ability $OA_{j}(Q_{j}, P_{1}, P_{2},\cdots, P_{n_{arx}})$ decreases.
\begin{prop}
\label{eq:sensorNneig}
For each anchor node $j$, the number of neighboring sensor nodes is higher or at least equal with the increase of transmission power $Q_{j}$.
\end{prop}
\begin{proof}
Let us assume that there are $n$ sensor nodes uniformly deployed in the area with size $d^{3}$. $n_{hd}(Q_{j})$ can be calculated by,
\begin{equation}
n_{hd}(Q_{j}) = \rho v_{j} = \frac{4\pi n R_{j}^{3}}{3d^{3}},
\end{equation}
where $\rho$ is the density of sensor nodes and $v_{j}$ is the volume of $j$th anchor node with the transmission range $R_{j}$.

According to Eqs. \eqref{eq:tp} and \eqref{eq:attenuation}, the transmission power $Q_{j}$ of anchor node $j$ is given as,
\begin{equation}
Q_{j}(R_{j})=A_{norm}R_{j}^{k}a(f)^{R_{j}} + Q_{j0}.
\end{equation}
The inverse function $g^{-1}(Q_{j})=R_{j}$ exists and is monotonically strictly increasing with the transmission power $Q_{j}$, because $\frac{\partial Q_{j}(R_{j})}{\partial R_{j}} = A_{norm}\left(kR_{j}^{k-1}a(f)^{R_{j}} + R_{j}^{k}a(f)^{R_{j}}\ln a(f)\right) > 0$. Therefore, $n_{hd}\left(Q_{j}\right)$ can be represented by Eq. \eqref{eq:nqj}. Then, the first order partial derivative of $n_{hd}(Q_{j})$ is given in Eq. \eqref{eq:neigProve},
\begin{equation}
\label{eq:nqj}
n_{hd}(Q_{j}) = \frac{4\pi n \left(g^{-1}(Q_{j})\right)^{3}}{3d^{3}},
\end{equation}
\begin{equation}
\label{eq:neigProve}
\frac{\partial n_{hd}(Q_{j})}{Q_{j}} = \frac{4\pi n}{d^{3}}\left(g^{-1}(Q_{j})\right)^{2}\frac{\partial g^{-1}(Q_{j})}{\partial Q_{j}}.
\end{equation}
Hence, $\frac{\partial n_{hd}(Q_{j})}{Q_{j}}>0$, which proves Proposition \ref{eq:sensorNneig}.
\end{proof}
Next, we define the payoff function of any anchor node $j$ by considering various factors such as energy cost, the ability to localize sensor nodes and the transmission power of sensor nodes as well as anchor nodes. It is hence defined as the weighted sum of the remaining energy ratio after transmission with $Q_{j}$ and its localization ability $OA_{j}(Q_{j}, P_{1}, P_{2},\cdots, P_{n_{arx}})$,
\begin{align}
\label{eq:anchorUtility}
&U_{F}\left(Q_{j}, P_{1}, P_{2},\cdots, P_{n_{arx}}\right)= w_{1j}\frac{E_{j}^{tl}-C_{j}Q_{j}}{E_{j}^{tl}}
\quad \nonumber \\& \quad \quad + w_{2j}OA_{j}\left(Q_{j}, P_{1}, P_{2},\cdots, P_{n_{arx}}\right).
\end{align}
In the first term of Eq. \eqref{eq:anchorUtility}, $E_{j}^{tl}$ is the total energy of the $j$th anchor node and $C_{j}$ is the transmission energy cost per unit power. Weights $w_{1j}$ and $w_{2j}$ define the trade-off between the energy consumption of anchor node and the localization ability of anchor node, and satisfy $w_{1j}+w_{2j}=1$, $w_{kj} \in (0, 1)$.
\subsection{Utility Function of Sensor Nodes}
In the considered localization problem, sensor nodes act as leaders. They watch for the decision of anchors which act as followers, and based on the response of followers, maximize their profits. The strategy of the leader is to minimize the energy consumption and to localize the maximum number of sensor nodes during the allowed localization delay, which is defined in Section \ref{subse:perMetric}. A sensor node broadcasts a `Request' message to explore the maximum number of anchors. After sensor nodes receive enough beacon locations from neighbor anchor nodes, it will localize itself.

The `ability of sensor node $i$ to find anchor nodes' $OS_{i}(P_{i}, Q_{1}, Q_{2},\cdots, Q_{n_{srx}})$ is composed of two terms,
\begin{align}
\label{eq:osability}
OS_{i}&(P_{i}, Q_{1}, Q_{2},\cdots, Q_{n_{srx}}) = \quad \nonumber \\&\frac{n_{neig}(P_{i})}{n_{srx}}- \frac{\sum_{j=1}^{n_{srx}}Q_{j}}{P_{i}}.
\end{align}
In Eq. \eqref{eq:osability}, the first term $\frac{n_{neig}(P_{i})}{n_{srx}}$ is the ratio of the number of anchor nodes $n_{neig}(P_{i})$ within `one-hop' and `two-hop' ranges of transmission power $P_{i}$ and the total number of anchor nodes $n_{srx}$ received, where $n_{neig}(P_{i})$ is a non-decreasing function of $P_{i}$, the proof of which is similar to that of Proposition~\ref{eq:sensorNneig}. The second term $\frac{\sum_{j=1}^{n_{srx}}Q_{j}}{P_{i}}$ is the ratio of the sum-transmission power of received anchor nodes and the transmission power of sensor node $i$. If the transmission power of anchor nodes $\sum_{j=1}^{n_{srx}}Q_{j}$ increases, more anchor nodes' beacon messages are received so less anchor nodes can be reached with a given power $P_{i}$. Therefore, the `ability of sensor node $i$ to find anchor nodes' $OS_{i}(P_{i}, Q_{1}, Q_{2},\cdots, Q_{n_{srx}})$ is inversely proportional to this ratio.

The payoff of any sensor node $i$ increases with the decrease in energy consumption. Also, it increases with the increase of the number of neighbor anchor nodes. In addition, the payoff of the leader decreases with each retry it does to send the `Request' message. Therefore, the payoff function of sensor node $i$ is defined as Eq.~\eqref{eq:unknownUtility}, which is the weighted sum of its remaining energy ratio and the `ability of sensor node $i$ to find anchor nodes',
\begin{align}
\label{eq:unknownUtility}
U_{L}&\left(P_{i}, Q_{1}, Q_{2},\cdots, Q_{n_{srx}}\right)= w_{1i}\frac{E_{i}^{tl}-E_{i}P_{i}}{E_{i}^{tl}} \quad \nonumber \\&+ w_{2i}OS_{i}(P_{i}, Q_{1}, Q_{2},\cdots, Q_{n_{srx}}).
\end{align}
In the first term of Eq. \eqref{eq:unknownUtility}, $E_{i}^{tl}$ is the total energy in the $i$th sensor node and $E_{i}$ is the transmission energy cost per unit power. Weights $w_{1i}$ and $w_{2i}$ provide a trade-off between energy consumption and `ability to find anchor nodes', satisfying $w_{1i} + w_{2i}= 1$, $w_{1i}, w_{2i} \in (0, 1)$.

\subsection{Existence and Uniqueness of Stackelberg Nash Equilibrium}
\label{subse:SNE}
The considered game achieves equilibrium when the sensor node (leader) selects the optimal transmission power to get its location with minimum energy consumption while anchor nodes (followers) choose their optimal transmission power to localize the maximum number of sensor nodes with minimum energy cost. At equilibrium, the benefit of each side can not be improved by unilaterally changing its own strategy. To find the Stackelberg equilibrium, each sensor node calculates the best reaction of anchor nodes to each of its mixed strategy and selects the mixed strategy that maximizes its own utility.

\subsubsection{Best Response Strategy of Anchor Nodes}
To define the strategy of the $j$th anchor node $Q_{j}$, the transmission power allocation problem can be cast as the optimization problem formulated below,
\begin{IEEEeqnarray*}{rCL}
\label{eq:brln}
    \max\limits_{Q_{j}}\ &U_{F}\left(Q_{j}, P_{1}, P_{2},\cdots, P_{n_{arx}}\right) = w_{1j}\frac{E_{j}^{tl}-C_{j}Q_{j}}{E_{j}^{tl}} \quad \nonumber \\&+ w_{2j}\left(\frac{n_{hd}(Q_{j})}{n_{arx}}+\frac{n_{hd}(Q_{j})}{\sum_{k=1}^{n_{arx}}n_{k}^{req}}-\frac{\sum_{i=1}^{n_{arx}}P_{i}}{Q_{j}}\right)\IEEEyesnumber
    \\s.t.\\
    & w_{1j}+w_{2j}=1,\quad w_{1j}, w_{2j} \in (0, 1), \IEEEyessubnumber*\\
    & Q_{j}\in[0, Q_{max}], P_{j}\in[0, P_{max}],\IEEEyessubnumber*\\
    & C_{j}, n_{arx}, n_{k}^{req}, E_{j}^{tl} > 0.\IEEEyessubnumber*
\end{IEEEeqnarray*}

All anchor nodes are non-cooperative. In Proposition \ref{th:existLocalized}, the existence of the best response strategy of each anchor node is proved and the unique equilibrium point is computed.
\begin{prop}
\label{th:existLocalized}
Let $Q_{j}$ be the strategy of the $j$th anchor node. The best response $Q_{j}^{\ast}$ of each anchor node is given as,
\begin{align}
\label{eq:brlocalized}
& Q_{j}^{\ast}(P_{1}, P_{2},\cdots, P_{n_{arx}}) = \quad \nonumber \\& \left(\frac{w_{2j}E_{j}^{tl}\sum_{i=1}^{n_{arx}}P_{i}}{w_{1j}C_{j}-Z_{j}(g^{-1}(Q_{j}))^{2}\frac{\partial g^{-1}(Q_{j})}{\partial Q_{j}}}\right)^{\frac{1}{2}},
\end{align}
with $w_{F}^{\ast}<w_{1j}<1$, $w_{1j}+w_{2j}=1$ and $w_{2j} \in (0, 1)$, where
\begin{align}
\label{eq:zj}
Z_{j}=\frac{4(1-w_{1j})\pi nE_{j}^{tl}}{d^3}\left(\frac{1}{n_{arx}}+\frac{1}{\sum_{k=1}^{n_{arx}}n_{k}^{req}}\right),
\end{align}
\begin{align}
\label{eq:wf}
w_{F}^{\ast}=\frac{\frac{Z_{j}}{1-w_{1j}}\left(g^{-1}(Q_{j})\right)^2\frac{\partial g^{-1}(Q_{j})}{\partial Q_{j}}}{\frac{Z_{j}}{1-w_{1j}}\left(g^{-1}(Q_{j})\right)^2\frac{\partial g^{-1}(Q_{j})}{\partial Q_{j}}+C_{j}}.
\end{align}
\end{prop}
\begin{proof}
We prove that problem \eqref{eq:brln} is a standard form of convex optimization and determine the expression of the optimum.

The first order partial derivative of $U_{F}\left(Q_{j}, P_{1}, P_{2},\cdots, P_{n_{arx}}\right)$ with respect to $Q_{j}$, for $j \in [1, N]$, is given as
\begin{align}
\label{eq:firstOrder}
    &\frac{\partial U_{F}\left(Q_{j}, P_{1}, P_{2},\cdots, P_{n_{arx}}\right)}{\partial Q_{j}}=\frac{-w_{1j}C_{j}}{E_{j}^{tl}}+w_{2j}\left({\frac{\sum_{i=1}^{n_{arx}}P_{i}}{Q_{j}^2}}
    \right.\quad \nonumber \\& \quad \quad\left.{+\frac{4\pi n}{d^3}\left(\frac{1}{n_{arx}}+\frac{1}{\sum_{k=1}^{n_{arx}}n_{k}^{req}}\right) (g^{-1}(Q_{j}))^{2}\frac{\partial g^{-1}(Q_{j})}{\partial Q_{j}}}\right).\quad \nonumber
\end{align}
Let $\frac{\partial U_{F}\left(Q_{j}, P_{1}, P_{2},\cdots, P_{n_{arx}}\right)}{\partial Q_{j}}=0$, then we get $Q_{j}^{2}$ as,
\begin{equation}
Q_{j}^{2}(P_{1}, P_{2},\cdots, P_{n_{arx}})=\frac{(1-w_{1j})E_{j}^{tl}\sum_{i=1}^{n_{arx}}P_{i}}{w_{1j}C_{j}-Z_{j}(g^{-1}(Q_{j}))^{2}\frac{\partial g^{-1}(Q_{j})}{\partial Q_{j}}},\quad \nonumber
\end{equation}
where $Z_{j}=\frac{4(1-w_{1j})\pi nE_{j}^{tl}}{d^3}\left(\frac{1}{n_{arx}}+\frac{1}{\sum_{k=1}^{n_{arx}}n_{k}^{req}}\right)$.

For existence of $Q_{j}$, the condition $w_{1j}C_{j}-Z_{j}(g^{-1}(Q_{j}))^{2}\frac{\partial g^{-1}(Q_{j})}{\partial Q_{j}}>0$ should be satisfied. Then, $Q_{j}^{\ast}(P_{1}, P_{2},\cdots, P_{n_{arx}})$ can be obtained by Eqs. \eqref{eq:brlocalized} and \eqref{eq:brln} with the condition $w_{F}^{\ast}<w_{1j}<1$, $w_{F}^{\ast}<w_{1j}<1$ and $w_{2j} \in (0, 1)$.

The second order partial derivative of $U_{F}\left(Q_{j}, P_{1}, P_{2},\cdots, P_{n_{arx}}\right)$ is given as
\begin{align}
&\frac{\partial^2 U_{F}\left(Q_{j}, P_{1}, P_{2},\cdots, P_{n_{arx}}\right)}{\partial Q_{j}^2}=
\frac{4\pi nw_{2j}}{d^3}\left({\frac{1}{n_{arx}}}\right.
\quad \nonumber \\&\quad \quad\left.{+\frac{1}{\sum_{k=1}^{n_{arx}}n_{k}^{req}}}\right)
\left({2g^{-1}(Q_{j})\left(\frac{\partial g^{-1}(Q_{j})}{\partial Q_{j}}\right)^2}\right.\quad \nonumber \\&\quad \quad\left.{+\left(g^{-1}(Q_{j})\right)^2\frac{\partial^2 g^{-1}(Q_{j})}{\partial Q_{j}^2}}\right)-\frac{2w_{2j}\sum_{i=1}^{n_{arx}}P_{i}}{Q_{j}^{3}}.\quad \nonumber
\end{align}

Here, we need to prove $G(Q_{j})<0$ in order to prove that $\frac{\partial^2 U_{F}\left(Q_{j}, P_{1}, P_{2},\cdots, P_{n_{arx}}\right)}{\partial Q_{j}^2}$ is negative, where
\begin{equation}
\label{eq:aa}
G(Q_{j})=\left(2g^{-1}(Q_{j})\left(\frac{\partial g^{-1}(Q_{j})}{\partial Q_{j}}\right)^2\left(g^{-1}(Q_{j})\right)^2\frac{\partial^2 g^{-1}(Q_{j})}{\partial Q_{j}^2}\right).
\end{equation}
Firstly, we prove $\frac{\partial^2 g^{-1}(Q_{j})}{\partial Q_{j}^2}<0$. According to Eqs. \eqref{eq:tp} and \eqref{eq:attenuation}, the transmission power $Q_{j}$ of anchor node $j$ is given as
\begin{equation}
Q_{j}(R_{j})=A_{norm}R_{j}^{k}a(f)^{R_{j}} + Q_{j0}.\quad \nonumber
\end{equation}
The inverse function $g^{-1}(Q_{j})=R_{j}$ of anchor node $j$ exists and is strictly monotonically increasing with the transmission power $Q_{j}$, because $\frac{\partial Q_{j}(R_{j})}{\partial R_{j}} = A_{norm}\left(kR_{j}^{k-1}a(f)^{R_{j}} + R_{j}^{k}a(f)^{R_{j}}\ln a(f)\right) > 0$ and $\frac{\partial Q_{j}(R_{j})}{\partial R_{j}}=\frac{1}{\frac{\partial g^{-1}(Q_{j})}{\partial Q_{j}}}$. Then, $\frac{\partial^2 g^{-1}(Q_{j})}{\partial Q_{j}^{2}}$ can be calculated as
\begin{align}
\label{eq:squareInverse}
\frac{\partial^2 g^{-1}(Q_{j})}{\partial Q_{j}^{2}}&=-\frac{\frac{\partial^2 Q_{j}(R_{j})}{\partial R_{j}^{2}}\frac{\partial g^{-1}(Q_{j})}{\partial Q_{j}}}{\left(\frac{\partial Q_{j}(R_{j})}{\partial R_{j}}\right)^2}\quad \nonumber \\&=-\frac{\partial^2 Q_{j}(R_{j})}{\partial R_{j}^{2}}\left(\frac{\partial g^{-1}(Q_{j})}{\partial Q_{j}}\right)^3.
\end{align}
Then, we have
\begin{align}
\label{eq:firstdOrder01}
&\frac{\partial^2 Q_{j}(R_{j})}{\partial R_{j}^{2}}=A_{norm}\left({k(k-1)R_{j}^{k-2}a(f)^{R_{j}}}\right.\quad \nonumber
\\&\quad  \left.{+kR_{j}^{k-1}a(f)^{R_{j}}\ln a(f)+kR_{j}^{k-1}a(f)^{R_{j}}\ln a(f)}\right.\quad \nonumber \\& \quad \left.{+R_{j}^{k}a(f)^{R_{j}}\ln^2 a(f)}\right)>0,
\end{align}
from which we get for \eqref{eq:squareInverse} $\frac{\partial^2 g^{-1}(Q_{j})}{\partial Q_{j}^{2}}<0$.

Secondly, we have $2g^{-1}(Q_{j})<\left(g^{-1}(Q_{j})\right)^2$, because $g^{-1}(Q_{j})=R_{j}$ and $R_{j}>>1$. In order to prove $G(Q_{j})<0$ given in Eq. \eqref{eq:aa}, we need to prove $\left(\frac{\partial g^{-1}(Q_{j})}{\partial Q_{j}}\right)^2<\left|\frac{\partial^2 g^{-1}(Q_{j})}{\partial Q_{j}^2}\right|$. According to Eq. \eqref{eq:squareInverse},
\begin{alignat}{2}
&\left|\frac{\partial^2 g^{-1}(Q_{j})}{\partial Q_{j}^2}\right|-\left(\frac{\partial g^{-1}(Q_{j})}{\partial Q_{j}}\right)^2\quad \nonumber \\
& = \left|\frac{\partial^2 Q_{j}(R_{j})}{\partial R_{j}^{2}}\left(\frac{\partial g^{-1}(Q_{j})}{\partial Q_{j}}\right)^3\right|-\left(\frac{\partial g^{-1}(Q_{j})}{\partial Q_{j}}\right)^2 \quad \nonumber \\
& = \left(\frac{\partial g^{-1}(Q_{j})}{\partial Q_{j}}\right)^2\left(\left|\frac{\partial^2 Q_{j}(R_{j})}{\partial R_{j}^{2}}\frac{\partial g^{-1}(Q_{j})}{\partial Q_{j}}\right|-1\right), \quad \nonumber
\end{alignat}
where $\frac{\partial^2 Q_{j}(R_{j})}{\partial R_{j}^{2}}$ and $\frac{\partial g^{-1}(Q_{j})}{\partial Q_{j}}$ are given as,
\begin{align}
\label{eq:con01}
&\frac{\partial^2 Q_{j}(R_{j})}{\partial R_{j}^{2}}=A_{norm}\left({k(k-1)R_{j}^{k-2}a(f)^{R_{j}}}\right.
\quad \nonumber \\ &\quad \quad\left.{+kR_{j}^{k-1}a(f)^{R_{j}}\ln a(f) + kR_{j}^{k-1}a(f)^{R_{j}}\ln a(f)}\right.
\quad \nonumber \\ &\quad \quad\left.{+ R_{j}^{k}a(f)^{R_{j}}\ln a(f)}\right),
\end{align}
\begin{equation}
\label{eq:con02}
\frac{\partial g^{-1}(Q_{j})}{\partial Q_{j}}=\frac{1}{A_{norm}\left(kR_{j}^{k-1}a(f)^{R_{j}} + R_{j}^{k}a(f)^{R_{j}}\ln a(f)\right)}.
\end{equation}
From Eq. \eqref{eq:highFre}, we know $\ln a(f)>>1$, from which we get $\frac{\partial^2 Q_{j}(R_{j})}{\partial R_{j}^{2}}\frac{\partial g^{-1}(Q_{j})}{\partial Q_{j}} \approx \frac{2k}{R_{j}}+1>1$, where $k=1.5$ and $R_{j}>0$. Therefore, $\left(\frac{\partial g^{-1}(Q_{j})}{\partial Q_{j}}\right)^2<\left|\frac{\partial^2 g^{-1}(Q_{j})}{\partial Q_{j}^2}\right|$ is proved.

Since the value of the second order partial derivative of $U_{F}\left(Q_{j}, P_{1}, P_{2},\cdots, P_{n_{arx}}\right)$ is negative, the maximum value of $U_{F}\left(Q_{j}, P_{1}, P_{2},\cdots, P_{n_{arx}}\right)$ can be achieved at $Q_{j}^{\ast}(P_{1}, P_{2},\cdots, P_{n_{arx}})$ by solving $\frac{\partial U_{F}\left(Q_{j}, P_{1}, P_{2},\cdots, P_{n_{arx}}\right)}{\partial Q_{j}}=0$, proving Proposition \ref{th:existLocalized}.
\end{proof}
\subsubsection{Best Response Strategy of Sensor Node}
To define the strategy of the $i$th sensor node $P_{i}$, the transmission power allocation problem can be cast as the optimization problem formulated below,
\begin{IEEEeqnarray*}{rCL}
\label{eq:brunln}
&\max\limits_{P_{i}}\ U_{L}\left(P_{i}, Q_{1}, Q_{2},\cdots, Q_{n_{srx}}\right) = w_{1i}\frac{\left(E_{i}^{tl}-E_{i}P_{i}\right)}{E_{i}^{tl}} \quad \nonumber \\&\quad \quad + w_{2i}\left(\frac{n_{neig}(P_{i})}{n_{srx}} - \frac{\sum_{j=1}^{n_{srx}}Q_{j}}{P_{i}}\right)\IEEEyesnumber
\\s.t.\\
& w_{1i}+w_{2i}=1,\quad w_{1i}, w_{2i} \in (0, 1), \IEEEyessubnumber*\\
& P_{i}\in[0, P_{max}], Q_{h}\in[0, Q_{max}], \IEEEyessubnumber*\\
& n_{srx}, E_{i}, E_{i}^{tl}> 0. \IEEEyessubnumber*
\end{IEEEeqnarray*}

The existence and uniqueness of the sensor node's transmission power at Nash equilibrium is proved in Proposition~\ref{th:existUnlocalized}.
\begin{prop}
\label{th:existUnlocalized}
Let $P_{i}$ be the strategy of the $i$th sensor node. The best response $P_{i}^{\ast}$ of each sensor node is given as,
\end{prop}
\begin{align}
\label{eq:brunlocalized}
&P_{i}^{\ast}(Q_{1}, Q_{2},\cdots, Q_{srx}) = \quad \nonumber \\& \quad \left(\frac{w_{2i}d^3E_{i}^{tl}n_{srx}\sum_{j=1}^{n_{srx}}Q_{j}}{w_{1i}A_{i}-Z_{i}\left(g^{-1}(P_{i})\right)^2\frac{\partial g^{-1}(P_{i})}{\partial P_{i}}}\right)^{\frac{1}{2}},
\end{align}
with $w_{L}^{\ast}<w_{1i}<1$, $w_{1i}+w_{2i}=1$ and $w_{2i} \in (0, 1)$,
where $A_{i}=d^3E_{i}n_{srx}$, $Z_{i}=4\pi nw_{2i}E_{i}^{tl}$ and
\begin{equation}
\label{eq:bb}
w_{L}^{\ast}=\frac{4\pi nE_{i}^{tl}\left(g^{-1}(P_{i})\right)^2\frac{\partial g^{-1}(P_{i})}{\partial P_{i}}}{d^3E_{i}n_{srx}+4\pi nE_{i}^{tl}\left(g^{-1}(P_{i})\right)^2\frac{\partial g^{-1}(P_{i})}{\partial P_{i}}}.
\end{equation}
\begin{proof}
We prove that problem \eqref{eq:brunln} is a standard form of convex optimization and determine the expression of the optimum.

The first order partial derivative of $U_{L}\left(P_{i}, Q_{1}, Q_{2},\cdots, Q_{n_{srx}}\right)$ with respect to $P_{i}$ is given as
\begin{align}
\label{eq:firstOrderun}
&\frac{\partial U_{L}\left(P_{i}, Q_{1}, Q_{2},\cdots, Q_{n_{srx}}\right)}{\partial P_{i}}=-\frac{w_{1i}E_{i}}{E_{i}^{tl}}+w_{2i}\left({\frac{4\pi n}{d^3n_{srx}}}\right.\quad \nonumber \\&\quad \quad \left.{(g^{-1}(P_{i}))^{2}\frac{\partial g^{-1}(P_{i})}{\partial P_{i}}+\frac{\sum_{j=1}^{n_{srx}}Q_{j}}{P_{i}^2}}\right).\quad \nonumber
\end{align}
Letting $\frac{\partial U_{L}\left(P_{i}, Q_{1}, Q_{2},\cdots, Q_{n_{srx}}\right)}{\partial P_{i}}=0$, we get
\begin{equation}
P_{i}^{2}(Q_{1}, Q_{2},\cdots, Q_{srx})=\frac{w_{2i}d^3E_{i}^{tl}n_{srx}\sum_{j=1}^{n_{srx}}Q_{j}}{w_{1i}A_{i}-Z_{i}\left(g^{-1}(P_{i})\right)^2\frac{\partial g^{-1}(P_{i})}{\partial P_{i}}},\quad \nonumber
\end{equation}
where $A_{i}=d^3E_{i}n_{srx}$ and $Z_{i}=4\pi nw_{2i}E_{i}^{tl}$.

To guarantee the existence of $P_{i}^{2}(Q_{1}, Q_{2},\cdots, Q_{srx})$, we need to set the condition $w_{1i}A_{i}-Z_{i}\left(g^{-1}(P_{i})\right)^2\frac{\partial g^{-1}(P_{i})}{\partial P_{i}}>0$.

The second order partial derivative of $U_{L}\left(P_{i}, Q_{1}, Q_{2},\cdots, Q_{n_{srx}}\right)$ is given as
\begin{align}
\label{eq:secondOrderun}
&\frac{\partial^2 U_{L}\left(P_{i}, Q_{1}, Q_{2},\cdots, Q_{n_{srx}}\right)}{\partial P_{i}^2}=\frac{4\pi nw_{2i}}{d^3n_{srx}}\left({2g^{-1}(P_{i})}\right. \nonumber \\& \quad \quad\left.{\left(\frac{\partial g^{-1}(P_{i})}{\partial P_{i}}\right)^2+\left(g^{-1}(P_{i})\right)^2\frac{\partial^2 g^{-1}(P_{i})}{\partial P_{i}^2}}\right) \nonumber\\
&\quad \quad-\frac{2w_{2i}\sum_{j=1}^{n_{srx}}Q_{j}}{P_{i}^3}.\quad \nonumber
\end{align}
Here, we use the fact that $G(Q_{j})<0$ given in Eq. \eqref{eq:aa} as proved for Proposition \ref{th:existLocalized}. Therefore, the value of the second order partial derivative of $U_{L}\left(P_{i}, Q_{1}, Q_{2},\cdots, Q_{n_{srx}}\right)$ is negative and the maximum value of $U_{L}\left(P_{i}, Q_{1}, Q_{2},\cdots, Q_{n_{srx}}\right)$ can be achieved at $P_{i}^{\ast}(Q_{1}, Q_{2},\cdots, Q_{srx})$, proving Proposition \ref{th:existUnlocalized}, for the existence and uniqueness of the Nash equlibrium of the proposed EELA game.

\end{proof}
\subsection{Algorithm Design}
We now design our proposed EELA algorithm, where sensor nodes localize themselves once they receive enough location beacon information from neighboring anchor nodes. The proposed algorithm consists of four phases.
\begin{itemize}
\item {\underline{Phase 1:} each sensor node builds a neighbor list containing the `Wakeup' (Type, ID, Time) message received from its neighbor anchor nodes. Each anchor node also builds its neighbor anchor list with its received `Wakeup' message from its neighbor anchor nodes. }
\item {\underline{Phase 2:} anchor nodes which received `Wakeup' messages from neighbor anchor nodes, broadcast their neighbor anchor list by using `AnchorNbr' (Type, ID, Time, NbrAnchorNodes) message. Each sensor node updates its neighbor list information and adds the anchor node's neighbor information. The game starts at the third phase of nodes communication.}
\item {\underline{Phase 3:} to start the opportunistic localization, each sensor node explores its maximum opportunities with the consideration of energy consumption and neighbor anchor nodes by the procedures described in Algorithm~\ref{sd:pseudoCode01}. `One-hop' neighbor anchor nodes are considered first. If the number of `one-hop' neighbor anchor nodes is enough to localize the sensor node, it will not handle the `two-hop' neighbors. This is because anchor nodes in `one-hop' neighbor list have more accurate information, such as the one-way time delay. However, due to the node mobility and random deployment, `two-hop' anchor nodes should be considered given the few `one-hop' anchor nodes in UWSNs. Fig. \ref{fi:twoHop} depicts the initial ¡®two-hop¡¯ transmission power calculation, where $AN_{1}$ and $AN_{2}$ (black circles), which are `one-hop' and `two-hop' anchor nodes respectively, act as multiple followers. $SN_{3}$ is the sensor node acting as the single leader. If $SN_{3}$ uses the maximum transmission power $P_{max}$ to have a transmission radius of $R_{max}$, $OA$ is the opportunistic localization range. Proposition \ref{th:twohop} is used to evaluate the transmission power required to reach the `two-hop' anchor nodes.
    \begin{prop}
    \label{th:twohop}
    Let $p_{31}$ and $p_{32}$ be the transmission powers of sensor node $SN_{3}$ required to send a Request message to anchor nodes $AN_{1}$ and $AN_{2}$, respectively. Let $q_{12}$ be the transmission power required at $AN_{1}$ to reach $AN_{2}$. Then, if anchor nodes $AN_{1}$ and $AN_{2}$ are in the `one-hop' and `two-hop' neighbor list of $SN_{3}$, respectively, we can set $p_{32} < p_{31} + q_{12}$. Moreover, node $AN_{2}$ in the `two-hop' neighbor list of $SN_{3}$ can be moved to the `one-hop' neighbor list, if we use $p_{3} = p_{31} + q_{12}$ as the transmission power of $SN_{3}$.
    \end{prop}
     \vspace{6pt}

\emph{Remark 1}: $q_{12}$ can be known at sensor node $SN_{3}$ by estimating the distance $d_{12}$ from the received anchor nodes' messages.

\emph{Remark 2}: The final optimal transmission power of the sensor node will be selected by Proposition \ref{th:existUnlocalized}}.
    \begin{proof}
         From Eq. \eqref{eq:tp}, the transmission power $P(d)$ is an increasing function of the distance $d$. Given the triangle inequality $d_{32} < d_{31} + d_{12}$, and since $q_{12}$ and $p_{31}$ are sufficient to cover distances $d_{12}$ and $d_{31}$ respectively, the `two-hop' neighbor anchor node $AN_{2}$ becomes a `one-hop' neighbor anchor node by setting $p_{3} = p_{31} + q_{12}$ as the new initial transmission power.


    \end{proof}
 If we have multiple `two-hop' anchor nodes, Proposition \ref{th:twohop} is applied sequentially, until the required number of nodes is reached.
    \begin{figure}[!t]
        \centering
        \includegraphics[scale=0.35]{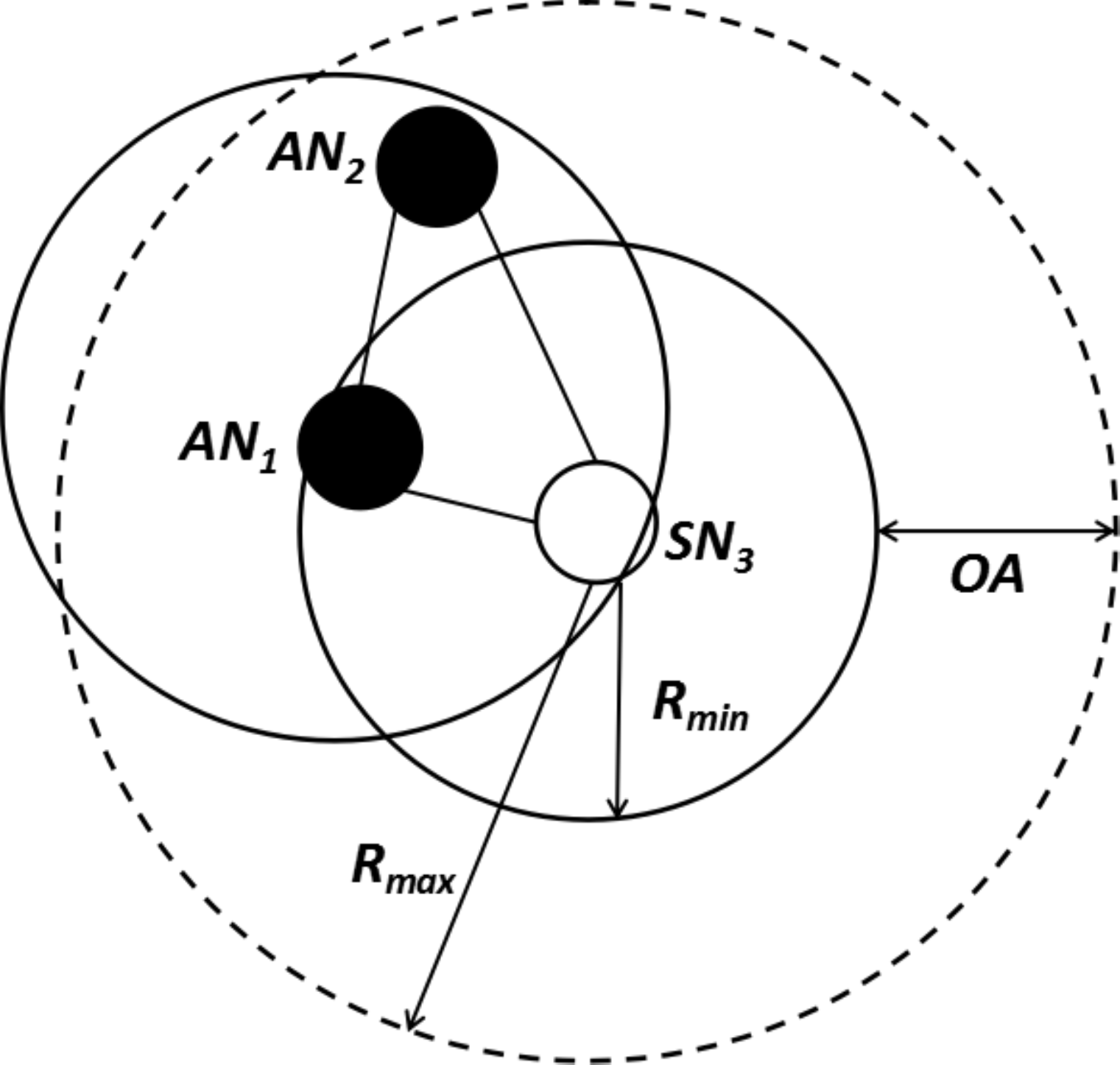}
        \caption{A scenario depicting the `two-hop' transmission power calculation}
        \label{fi:twoHop}
    \end{figure}
    \begin{algorithm}
        \SetKwInOut{Input}{Input}
        \SetKwInOut{Output}{Output}
        \Input{$\left\{N_{s}\right\}$, $\left\{N_{a}\right\}$, $P_{ini}$, $\{Nbr_{i}^{s}\}$, $n_{min}^{req}$}
        \Output{Optimized action $P_{i}{'}$}
         $P_{i}^{'} = 0$, $P_{i} = P_{ini}$, $U_{i}^{L\ast} = -\infty$.\\
         \For{each message received from an anchor node $j \in \{Nbr_{i}^{s}\}$}
         {    \eIf{$j \in \{Nbr_{i}^{s}\}^{`one-hop'}$}
              {
                 Add $j$ to $\mathcal{V}_{i}^{`one-hop'}$.
              }
              {
                Add $j$ to $\mathcal{V}_{i}^{`two-hop'}$.
              }
         }
        \eIf{$|\mathcal{V}_{i}^{`one-hop'}| \geq n_{min}^{req}$}
          {
            \For{each anchor node $j \in \mathcal{V}_{i}^{`one-hop'}$}{
                Calculate the utility $U_{i}^{L}$.
            }
            $P_{i}^{'} \gets argmax\ U_{i}^{L}$.\\
          }
          {
           \eIf{$0 \leq (|\mathcal{V}_{i}^{`one-hop'}| + |\mathcal{V}_{i}^{`two-hop'}|) <  n_{min}^{req}$}
           {
               $n_{i}^{req} \gets (n_{min}^{req} - |\mathcal{V}_{i}|)$,\\
               $P_{i}^{'} \gets P_{max}$.
           }
           {
              \For{each anchor node $j \in \mathcal{V}_{i}$}{
                Calculate the utility $U_{i}^{L}$.
              }
              $P_{i} \gets argmax\ U_{i}^{L}$.\\
              \eIf{$P_{i} > P_{max}$}
              {
                $P_{i}^{'} \gets P_{max}$.
              }
              {
                $P_{i}^{'} \gets P_{i}$.
              }
           }
        }
        Broadcast `Request' message at transmission power $P_{i}^{'}$.\\
        \caption{Topology control of a sensor node (leader)}
        \label{sd:pseudoCode01}
    \end{algorithm}
\item {\underline{Phase 4:} After anchor nodes receive the `Request' (Type, ID, Time, $n_{req}$) message from sensor nodes, an optimal transmission power will be selected by Proposition \ref{th:existLocalized} to broadcast the `Reply' (Type, ID, Time, Location) message taking into account the factors of energy consumption and the ability to localize sensor node. The detailed steps are given in Algorithm 2.
\begin{algorithm}
    \SetKwInOut{Input}{Input}
    \SetKwInOut{Output}{Output}
    \Input{$\left\{N_{a}\right\}$, $\left\{N_{s}\right\}$, $Q_{ini}$, $\{Nbr_{j}^{a}\}$}
    \Output{Optimized action $Q_{j}{'}$}
     $Q_{j}^{'} = 0$, $Q_{j} = Q_{ini}$.\\
     Broadcast `Wakeup' message at transmission power $Q_{j}$.\\
     \For{each `Wakeup'  message received from each anchor node}
     {    Build its neighbor anchor list $\{N_{j}^{a}\}$.
     }
     Broadcast `AnchorNbr' message at the transmission power $Q_{j}$.\\
     \For{each `Request'  message received from each sensor node $g \in \{Nbr_{j}^{a}\}$}
     {    Calculate the utility $U_{F}^{j}$.
     }
     $Q_{j}^{'} \gets argmax\ U_{F}^{j}$.\\
     Broadcast `Reply' message at transmission power $Q_{j}^{'}$.
    \caption{Topology control of an anchor node (follower)}
    \label{sd:pseudoCode02}
\end{algorithm}
}
\end{itemize}
Finally, after sensor nodes receive the required number of beacon location information from its neighboring anchor nodes, they execute their localization procedure. Since the main purpose of this paper is energy efficiency improvement of localization by topology control, we will assume the trilateration technique \cite{manolakis1996efficient} in the next section, for node localization to illustrate the proposed EELA. In trilateration technique, each sensor node requires three anchor nodes in order to obtain its location, i.e., $n_{min}^{req}=3$.

\section{Numerical Evaluations}
\label{sec:Simulations}
\subsection{Simulation Settings}
We consider a network of 10-50 sensor nodes in a 3D underwater region of $2500m \times 2500m \times 2500m$ with four anchor nodes on the water surface. In our simulation, the transmission range $R_{i}$ is a continuous number in $(0,\ max\_range]$. We set the $max\_range=\sqrt{2500^{2}+2500^{2}+2500^{2}}m$ which is the max distance in the simulation region. Initially, the value of all nodes' transmission range are $R_{ini}$. For each simulation, Underwater Acoustic Network (UAN) models of NS-3 are utilized for generating the channels, PHY and MAC layers, as they are commonly used for modeling underwater networks~\cite{parrish2008system}. The transmission power $P_{i}$ or $Q_{j}$ (TxPower attribute in UanPhyGen model) is initially set for a given range $R_{i}$ or $R_{j}$ using the Thorp's propagation model \cite{brekhovskikh1982scattering}. Weights of anchor nodes' utility function are taken as $w_{1j} = 0.4$ and $w_{2j} = 0.6$. On the other hand, sensor nodes' weights $w_{1i}$ and $w_{2i}$ are set to 0.1 and 0.9. This is because maximizing the ability of finding anchor nodes is crucial, while setting $w_{1i}=0.1$ for energy consumption still ensures very high energy-efficiency, and as will show by the simulation results. Other simulation parameters are listed in Table \ref{tb:paras}.
\begin{table}[t]
  \centering
    \begin{tabular}{ll}
		\hline
		Parameter & Value \\
		\hline
        Node mobility model & Meandering current mobility model \cite{caruso2008meandering} \\
        Channel Frequency & $22\ kHz$ \cite{misra2015game}\\
        Modulation technique & $FSK$ \cite{misra2015game}\\
        Data rate & $500\ bps$ \cite{misra2015game}\\
        Speed of sound & $1500\ m/s$ \cite{misra2015game}\\
        Wave propagation model & Thorp's propagation model \cite{brekhovskikh1982scattering}\\
        Receive and Idle power & $0.1\ watts$ \\
        Sleep power & $1 \times 10^{-4}\ watts$ \\
		\hline
    \end{tabular}
    \caption{Simulation parameters}
    \label{tb:paras}
\end{table}

In the simulations, sensor nodes are randomly deployed under water while anchor nodes are randomly deployed on the water surface in the simulation area. Any sensor node gets localized after receiving $n_{min}^{req}$ number of replies from anchor nodes. All nodes move according to the velocity of ocean current, following the Meandering Current Mobility (MCM) model \cite{caruso2008meandering}. In MCM, the effect of the meandering sub-surface currents and vortices are considered for nodes moving. The sensor nodes mobility ($v_{m}$) is set to 2.0m/s, 3.0m/s, 4.0m/s \cite{misra2015game}. Each simulation runs for 5000 times to obtain the average results.

\subsection{Baseline Schemes}
In the proposed EELA scheme, both anchor and sensor nodes can select their optimal transmission range to communicate with each other. The performance of the proposed scheme is compared to the three schemes listed below.
\begin{enumerate}
\item [1)] OLTC \cite{misra2015game}: an anchor node dynamically selects a transmission range from $[0, Q_{max}]$ in order to maximize the number of neighbors yet to be localized. Only anchor nodes can adjust their transmission range while sensor nodes always use the maximum transmission range to send messages.
\item [2)] EELA-Min: the scheme without the dynamic transmission power optimization, i.e., both anchor and sensor nodes use the fixed minimum transmission range to broadcast message.
\item [3)] EELA-Max: the scheme without the dynamic transmission power optimization, i.e., both anchor and sensor nodes use the fixed maximum transmission range to broadcast message.
\end{enumerate}
\subsection{Performance Metrics}
\label{subse:perMetric}
The following metrics are adopted to evaluate the performance of EELA.
\begin{enumerate}
\item [1)] Localization coverage: the ratio of the number of localized sensor nodes to the total number of sensor nodes,
    \begin{equation}
    \label{eq:coverage01}
        C=\frac{|\left\{N_{s\_l}\right\}|}{|\left\{N_{s}\right\}|}, \quad \nonumber \\
    \end{equation}
    where $\left\{N_{s\_l}\right\}$ is the set of sensor nodes which already have obtained their locations and the set of sensor node is $\left\{N_{s}\right\}$.
\item [2)] Average energy consumption per sensor node: the ratio of total energy consumption of sensor nodes to the number of sensor nodes,
    \begin{equation}
    \label{eq:AveSnNode}
        \varepsilon_{\left\{N_{s\_l}\right\}}^{avg}=\frac{\sum_{i=1}^{|\left\{N_{s\_l}\right\}|}\varepsilon_{i}}{|\left\{N_{s\_l}\right\}|}. \quad \nonumber \\
    \end{equation}
\item [3)] Average energy consumption per anchor node: the ratio of total energy consumption of anchor nodes to the number of anchor nodes,
    \begin{equation}
    \label{eq:AveAnNode}
        \varepsilon_{\left\{N_{a}\right\}}^{avg}=\frac{\sum_{j=1}^{|\left\{N_{a}\right\}|}\varepsilon_{j}}{|\left\{N_{a}\right\}|}, \quad \nonumber \\
    \end{equation}
    where the set of anchor node is $\left\{N_{a}\right\}$.
\item [4)] Average energy consumption per node: the ratio of the total energy consumption of all nodes to the number of all nodes,
    \begin{equation}
    \label{eq:AveperNode}
        \varepsilon_{tl}^{avg}=\frac{\sum_{i=1}^{|\left\{N_{s}\right\}|}\varepsilon_{i}+\sum_{j=1}^{|\left\{N_{a}\right\}|}}{|\left\{N_{s}\right\}|+|\left\{N_{a}\right\}|}. \quad \nonumber \\
    \end{equation}
\item [5)] Average localization error,
    \begin{equation}
    \label{eq:AveLoError}
        \frac{\sum_{i=1}^{|\left\{N_{s\_l}\right\}|}\sqrt{(x_{i}-x_{i}^{'})^{2}+(y_{i}-y_{i}^{'})^{2}+(z_{i}-z_{i}^{'})^{2}}}{|\left\{N_{s\_l}\right\}|}. \quad \nonumber \\
    \end{equation}
  where for any localized sensor node node $i$, $(x_{i},y_{i},z_{i})$ and $(x_{i}',y_{i}',z_{i}')$ denote the original and the estimated locations, respectively.
\item [6)]Average localization delay: the time duration from a sensor node broadcasting a `Request' message to the time of obtaining its location.\\
\end{enumerate}
\subsection{Results and Analysis}

\subsubsection{Localization Coverage}
In Fig. \ref{fi:coverage}, the average localization coverage in function of the number of sensor nodes is presented.

We observe that the average localization coverage of both EELA and OLTC are between that of EELA-Max and EELA-Min, as expected. The proposed EELA scheme outperforms the reference OLTC scheme of \cite{misra2015game}. For EELA and OLTC, an increased number of sensor nodes results in a better localization coverage. This is because the increase of the number of sensor nodes entails a higher spatial density, hence more sensor nodes may be localized by anchor nodes with a given power. For EELA-Max, the number of successfully localized sensor nodes remains the same, since the maximum allowed transmission power is always used.

Compared to OLTC, the localization coverage achieved by EELA is $2.0\%$ higher on average. This is because in OLTC scheme, sensor nodes always send request with the maximum transmission power, which leads to a higher rate of packet collision. Thus, anchor nodes will receive fewer `Request' messages. However, anchor nodes use the optimal transmission power $P_{j} < P_{max}$ to reply, so that some sensor nodes may not receive the required number of anchor nodes, hence decreasing coverage. By contrast, in EELA, sensor nodes request with the optimal transmission power $P_{j} < P_{max}$ instead of using the maximum transmission power, so that the required number of anchor nodes may be reached with minimum energy consumption. On the other hand, anchor nodes also utilize the optimal transmission power to reply. Both the optimal transmission power for anchor node and sensor node are selected to reach Stackelberg Nash Equilibrium in section \ref{subse:SNE}. Hence, both anchor nodes and sensor nodes cannot improve their individual profit by single-sidedly changing their transmission power. Therefore, EELA can reduce the rate of packet collision, and achieve a higher coverage.

Compared to EELA-min, the localization coverage achieved in EELA is about 54\% higher on average, because sensor nodes always use the minimum transmission power to send `Request' messages while anchor nodes also use the minimum transmission power to reply, hence each sensor node receives the fewest beacon location messages to localize itself. Next, the localization coverage achieved in EELA-max is about 9\% (on average) higher than that in EELA. For EELA-max, although the `Request' messages from sensor nodes have higher probability to have collisions, anchor nodes always use the maximum transmission power to send `Reply' messages without considering received requests. Therefore, sensor nodes can always receive more beacon locations than those of other schemes. That is the reason why EELA-max always has the highest coverage. However, the use of higher transmission power leads to higher energy consumption, as shown next.
%
\begin{figure}
\center
\includegraphics[scale=0.45]{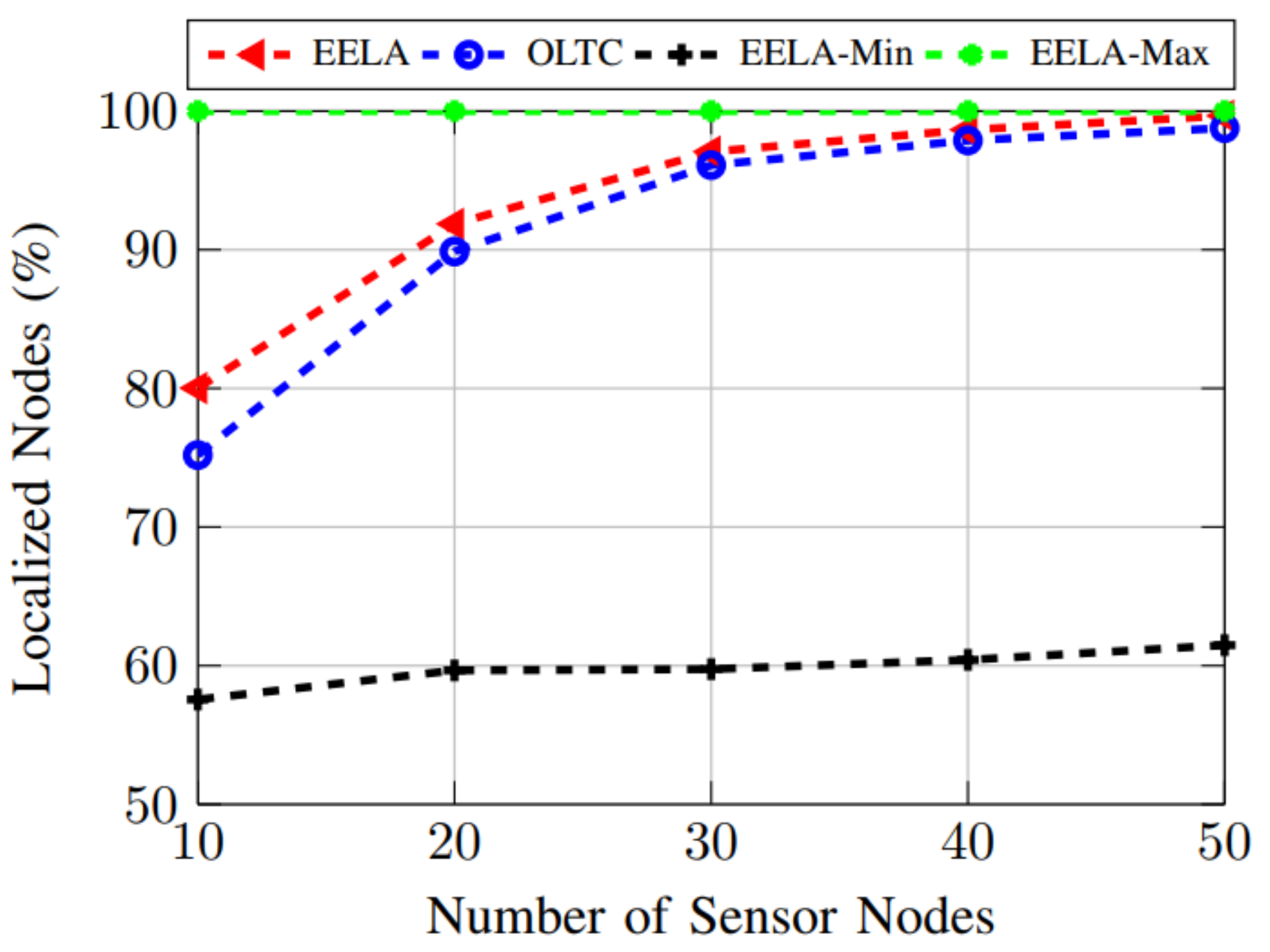}
\caption{Localization coverage}
\label{fi:coverage}
\end{figure}

\subsubsection{Average Energy Consumption Per Sensor Node}
\label{subsubse:AveLoNode}
The average energy consumption results for sensor nodes are given in Fig. \ref{fi:AveUnEny}.

We observe that the performance of the proposed EELA scheme is between that of EELA-Min and EELA-Max, while OLTC has the same consumption as EELA-Max since all sensor nodes transmit with maximum power. In OLTC, EELA-Min and EELA-Max, the average energy consumption is not affected by the node density of sensor nodes because of the fixed transmission power. The variations for proposed EELA are also steady due to the strategy of the proposed game and the constant number of anchor nodes in this simulation setting.

From Fig. \ref{fi:AveUnEny}, we notice that the average energy consumption per sensor node in EELA (326J) is about 53\% lower than that in OLTC (693J). This is thanks to our transmission power optimization strategy given the `one-hop' and `two-hop' nodes, which enables to reach the same number of anchor nodes as by OLTC, but with much lower energy. This mechanism significantly reduces the energy consumption of sensor nodes, and reduces the collisions of requests at the same time.
\begin{figure}
\center
\includegraphics[scale=0.45]{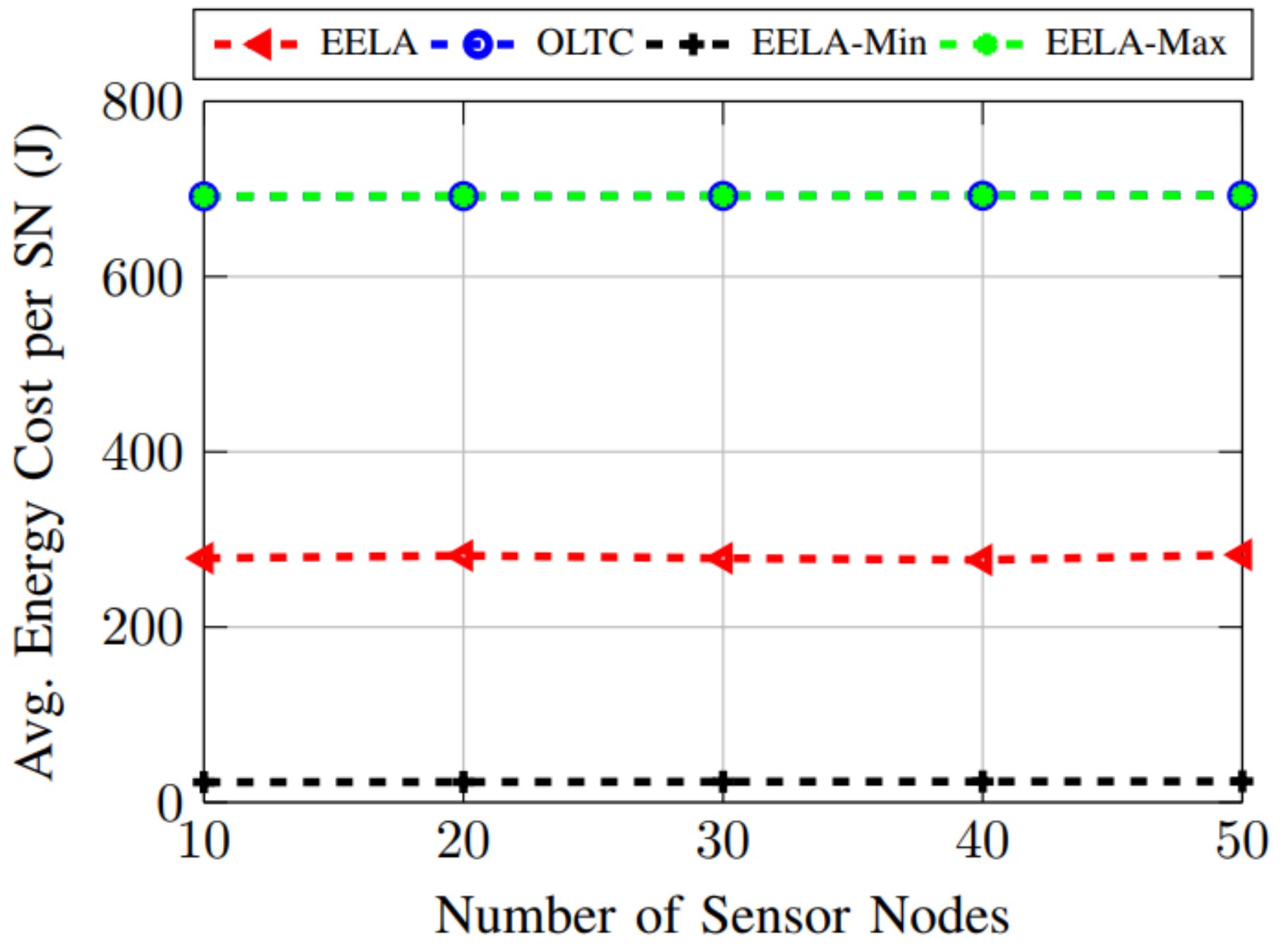}
\caption{Average energy consumption per Sensor Node (SN)}
\label{fi:AveUnEny}
\end{figure}
\subsubsection{Average Energy Consumption Per Anchor Node}
\label{subsub:AvePerAN}
Fig.~\ref{fi:AveAnEny} illustrates the average energy consumption of anchor nodes in function of the number of sensor nodes.

We observe that the energy cost of proposed EELA and OLTC lies between that of EELA-min and EELA-max, with a higher consumption for EELA compared to OLTC. Namely, we can see that the average energy consumption per anchor node in EELA (about 407J) is nearly 38\% higher than that in OLTC (about 295J). This is because anchor nodes in EELA need to broadcast twice in order to build their `two-hop' anchor neighboring list. As for OLTC, anchor nodes do not need to consider about their `two-hop' anchor neighboring nodes. As shown next, EELA slightly increases the energy consumption of anchor nodes in order to improve the performance of sensor nodes, which eventually improves the average energy consumption of all nodes. Note also that saving energy of underwater sensor nodes is a more crucial issue than that for anchor nodes, since anchor nodes are specifically deployed at the surface for enabling localization.
%
\begin{figure}
\center
\includegraphics[scale=0.45]{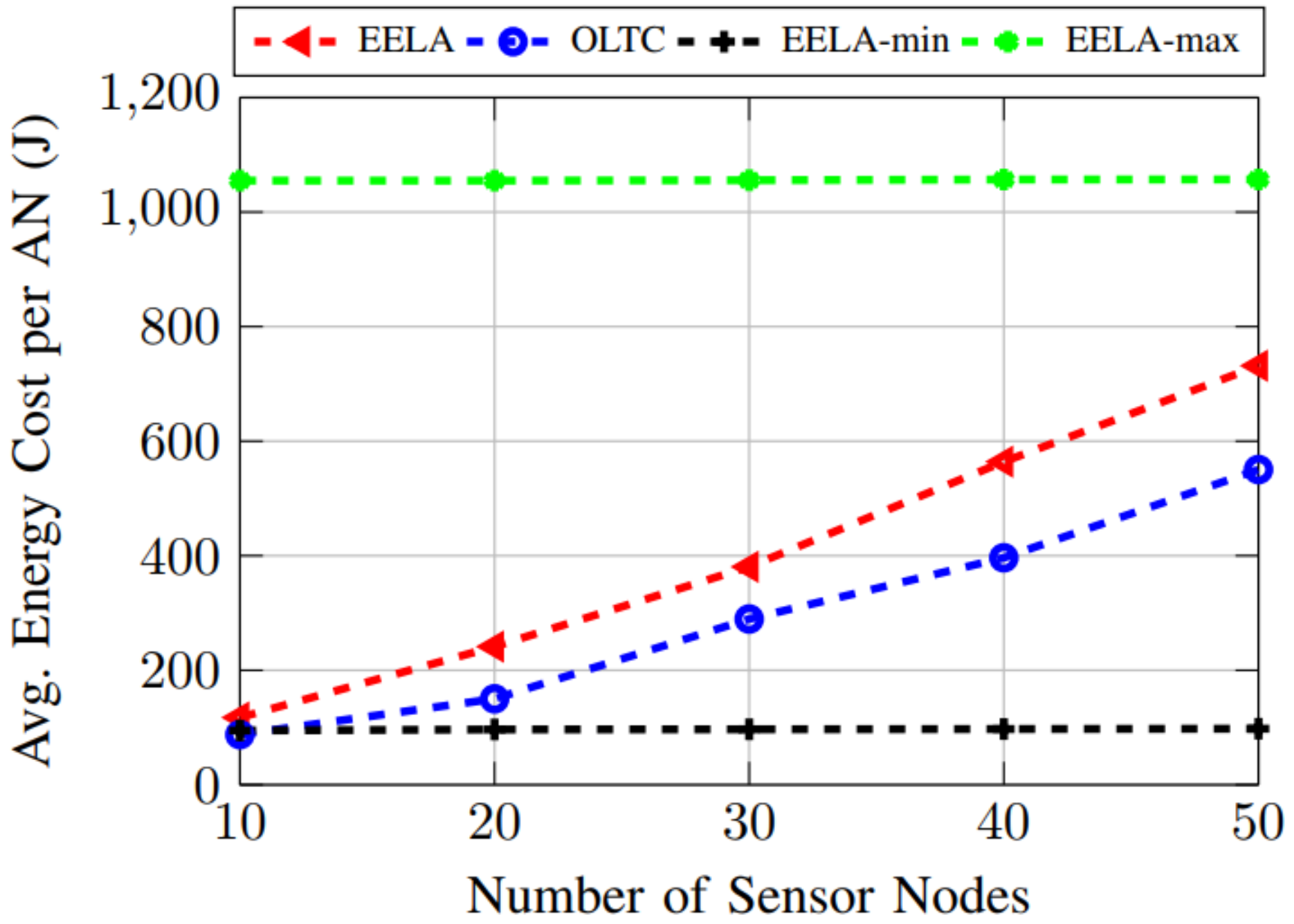}
\caption{Average energy consumption per Anchor Node (AN)}
\label{fi:AveAnEny}
\end{figure}
\subsubsection{Average Energy Consumption per Node}
\label{subsubse:AvePerNode}
Fig. \ref{fi:AveToEny} shows the average energy consumption over anchor and sensor nodes.

Overall, we observe that the proposed scheme largely reduces the average energy cost per node, compared to OLTC, i.e., around 48\% reduction. This shows that even if anchor nodes broadcast twice in the preprocessing phase, the proposed EELA still consumes much less energy in total, thanks to the energy-efficient power selection of sensor nodes. This is because in the deployment of practical localization systems, the number of sensor nodes is much larger than that of anchor nodes.

\begin{figure}
\center
\includegraphics[scale=0.45]{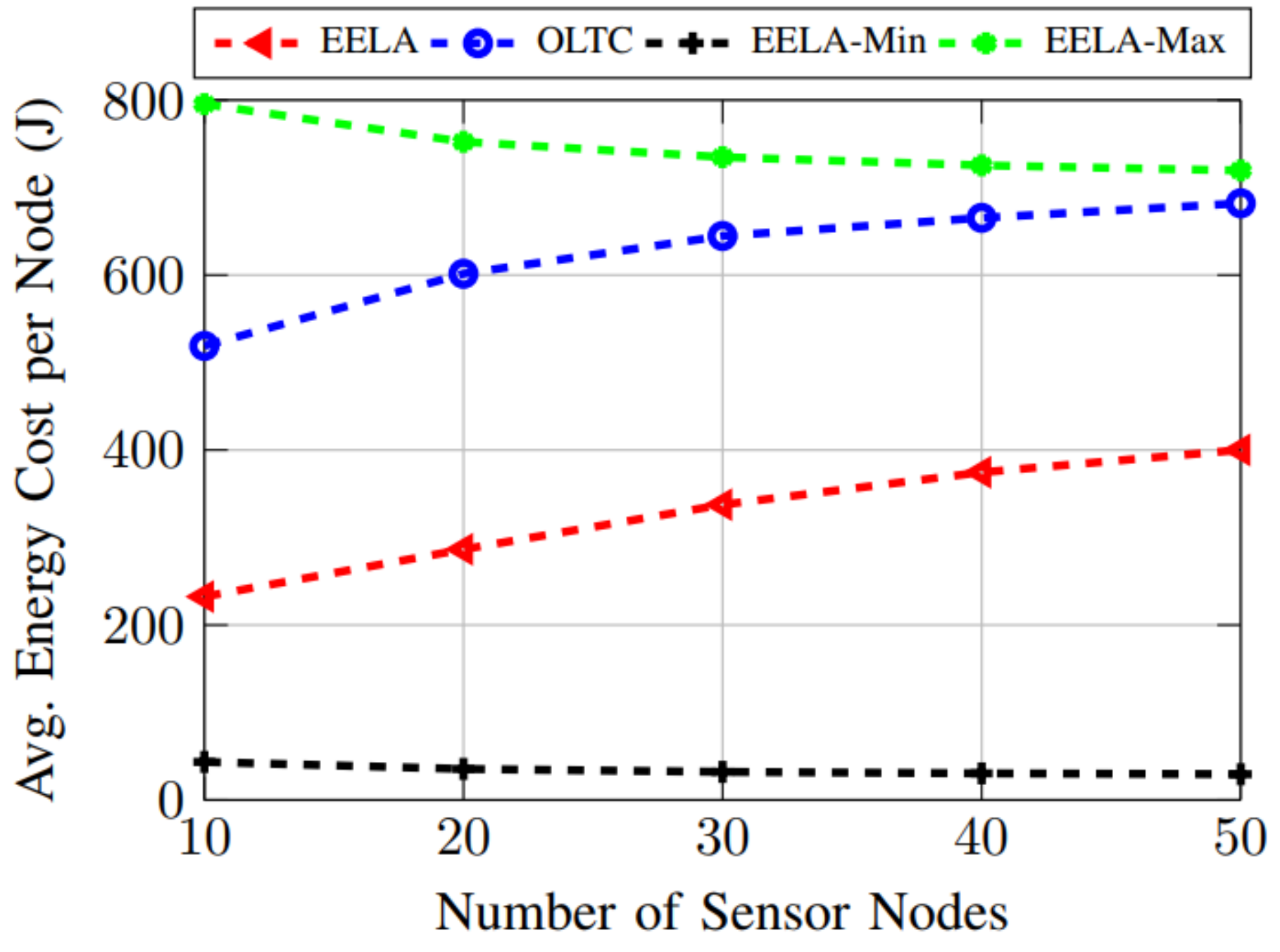}
\caption{Average energy consumption per node (anchor and sensor)}
\label{fi:AveToEny}
\end{figure}

%
\subsubsection{Average Localization Delay and Error}
Table \ref{tb:localDelayError} represents the Average Localization Delay (ALD) and Average Localization Error (ALE) of different schemes.

We notice that the ALD of EELA is almost the same as that of OLTC. The ALD of EELA-Min is lower than that of EELA by nearly 11\% on average. This is because in EELA-Min, the communication distance travelled by the acoustic signal is shorter than that in other schemes, since it uses the minimum transmission power. As for EELA-Max, the communication distance is longer than that in other schemes since it uses maximum transmission power. 

Next, for evaluating ALE, each sensor node requires 3 beacon locations and 3 distances from anchor nodes since the trilateration technique is considered. It is assumed that anchor nodes broadcast their precise coordinates, so the localization error is generated by the mobility of nodes and depends on the distance between anchor and sensor nodes. From Table \ref{tb:localDelayError}, we can see that EELA performs slightly better than OLTC, while the lowest and highest ALE are achieved by EELA-min and EELA-max, respectively. Overall, the ALEs are at comparable and reasonable levels for all algorithms.


\begin{table}[t]
  \centering
    \begin{tabular}{ccccc}
		\hline
		Metric & EELA & OLTC & EELA-Min & EELA-Max\\
		\hline
        ALD (s) & 6.87 & 6.85 & 6.11 & 7.15 \\
        ALE (m) & 3.18 & 3.24 & 3.17 & 3.30\\
		\hline
    \end{tabular}
    \caption{Average localization delay and error}
    \label{tb:localDelayError}
\end{table}
\subsection{Environmental Influences in Simulation}
In this section, we consider different realistic values of the current speeds ($v_{m} = 2.0,\ 3.0,\ 4.0\ m/s$) to observe the effects of environment changes.

Table \ref{tb:APDS} gives the average performance of EELA with different speed, in terms of the Average Localization Coverage (ALC), Average Energy per Node (AEN) and ALD. We observe that the performance of proposed EELA is rather constant under different speeds, showing the stability of EELA against current variations.
\begin{table}[t]
  \centering
    \begin{tabular}{cccc}
		\hline
		Metric & $2m/s$ & $3m/s$ & $4m/s$\\
		\hline
        ALC (\%)& 96.24 & 96.16 & 96.26 \\
        AEN (J) & 249.15 & 249.44 & 247.46 \\
        ALD (s) & 6.90 & 6.89 & 6.89 \\
		\hline
    \end{tabular}
    \caption{Average performance of EELA under various current speeds}
    \label{tb:APDS}
\end{table}

Finally, Fig. \ref{fi:AveLoErrorSpeed} shows that, the higher the current velocity (from 2m/s to 4m/s), the higher the average localization error (from 3.2m to 8.4m). This is due to the delay between the time where anchor nodes' broadcast  their messages used for location estimation, and the actual. For example, if this transmission delay is between 0 to 3 seconds, the node may move 0-6m. When the current speed is 2m/s, it is reasonable that the average transmission error is around 3.2m. Similarly, when the speed of current is 4m/s, the average localization error is around 8.4m. However, the average localization errors remain constant with the number of sensors, which also shows the stability of the proposed EELA scheme in terms of ALE.

\begin{figure}
\center

\includegraphics[scale=0.45]{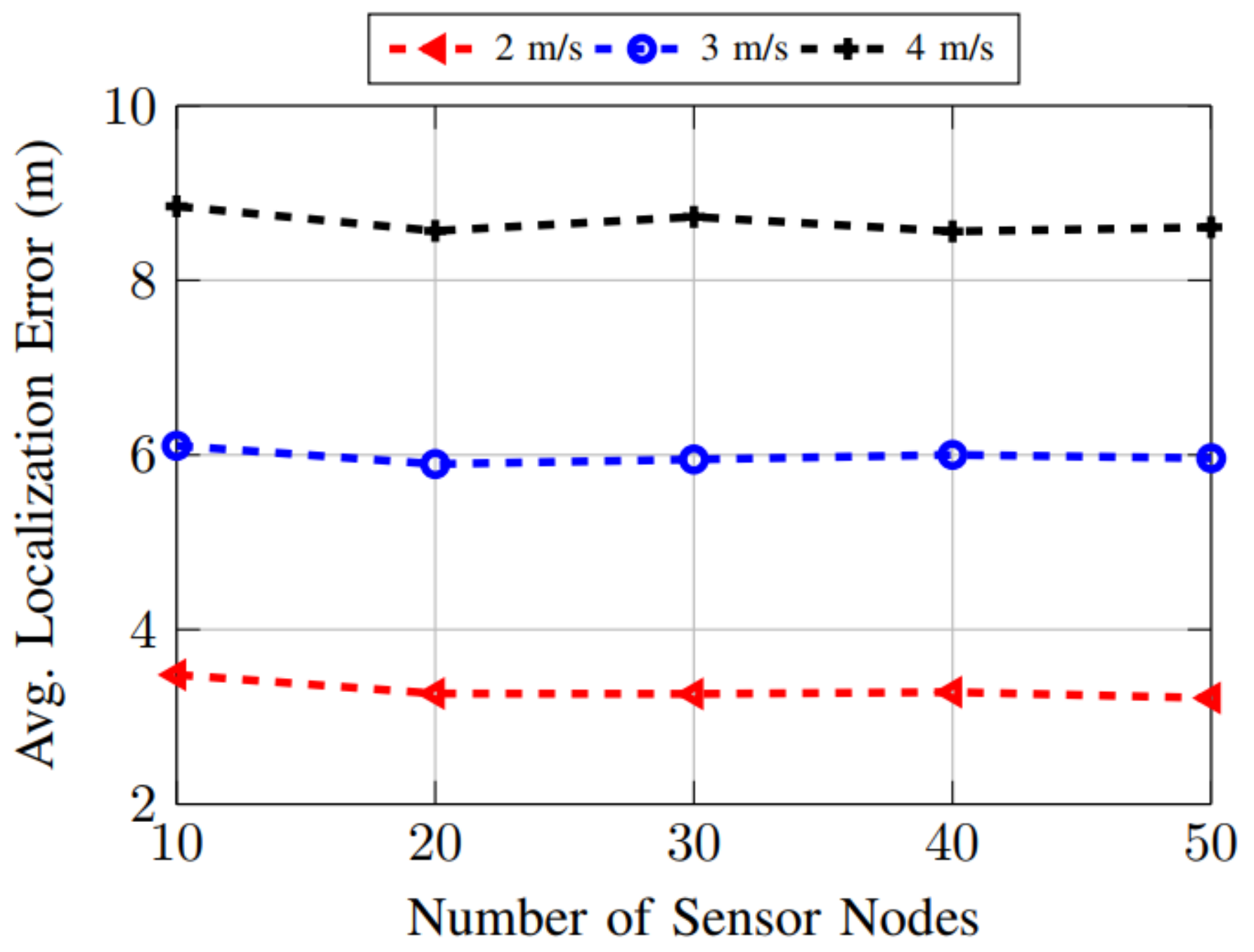}

\caption{Average localization error under various current speeds}
\label{fi:AveLoErrorSpeed}
\end{figure}
\section{Conclusion}
\label{sec:Conclusion}
We have considered the problem of energy-efficient sensor node localization using multiple anchor nodes, in underwater sensor networks where battery saving is essential. A Single-Leader-Multi-Follower Stackelberg game is used to model the considered localization problem, where anchor nodes act as followers of each sensor node, which acts as a leader. Considering the trade-off between localization ability and energy consumption, optimal transmission power strategies are devised for anchor and sensor nodes, which are shown to achieve Nash Equilibrium. Based on this analysis, we have proposed the EELA algorithm defining the communication protocol among anchor and sensor nodes, for enabling energy-efficient localization. Simulation results demonstrate that compared to baseline schemes, the proposed EELA enables similar or better performance in terms of localization coverage, errors and delays, while drastically reducing the amount of consumed energy, i.e., down to half the consumption of reference OLTC.

In the future work, we will enhance the proposed method to cope with the dynamic variations of the environment, by incorporating learning strategies.
\section*{Acknowledgment}
Yali Yuan would like to thank the scholarship support from the China Scholarship Council (CSC).

This work was supported by the Grants-in-Aid (JSPS Kakenhi) for Scientific Research no. 17K06453 from the Ministry of Education, Science, Sports, and Culture of Japan, and by the NII MoU grants.

\bibliographystyle{ieeetr}
\bibliography{TVTEELM}
\begin{IEEEbiography}[{\includegraphics[width=1in,height=1.25in,clip,keepaspectratio]{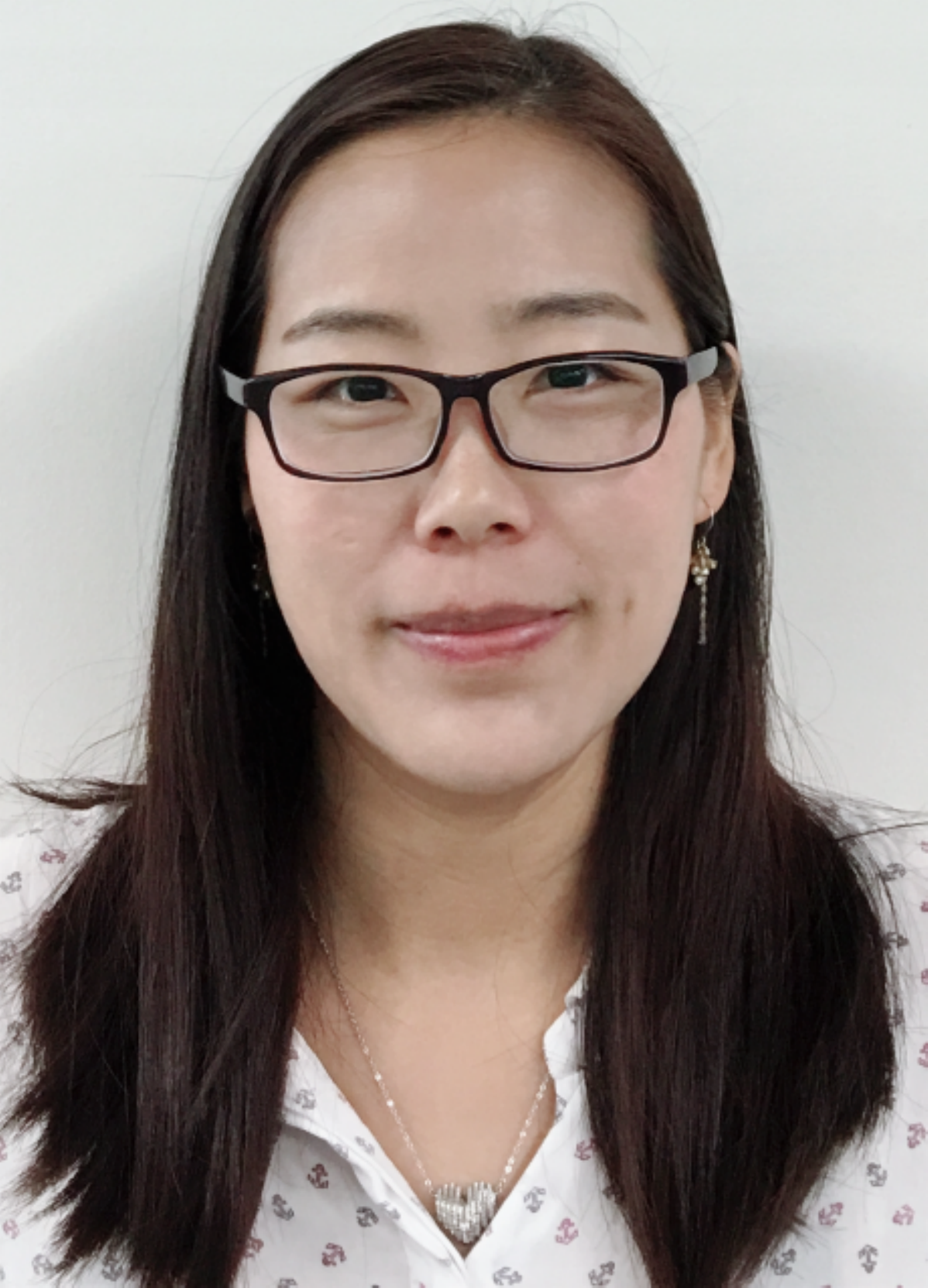}}]{Yali Yuan}
received the M.Sc. degree from University of Lanzhou, Lanzhou, China, in 2015. She is currently working towards her Ph.D. degree in Telematics Group at the Institute of Computer Science, University of G\"ottingen, G\"ottingen, Germany. Her research interests include various topics related to wireless networks and machine learning, in particular for the localization and security.
\end{IEEEbiography}
\begin{IEEEbiography}[{\includegraphics[width=1in,height=1.25in,clip,keepaspectratio]{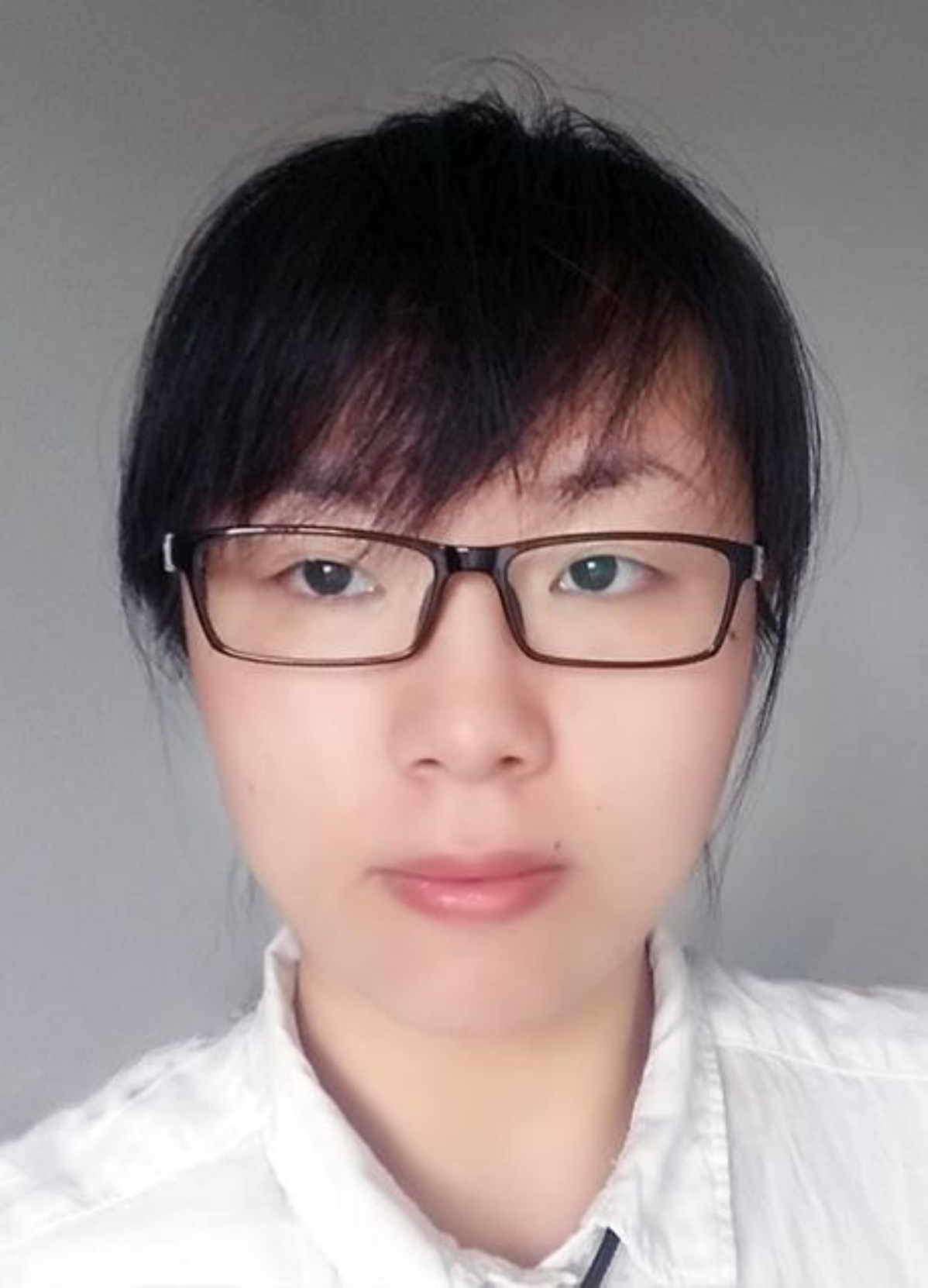}}]{Chencheng Liang}
received the B.S. degree in software engineering from University of Chengdu, China, in 2014. She is currently working towards her M.Sc. degree in Telematics Group at the Institute of Computer Science, University of G\"ottingen, Germany. Her research interests include wireless networks and machine learning.
\end{IEEEbiography}
\begin{IEEEbiography}[{\includegraphics[width=1in,height=1.25in,clip,keepaspectratio]{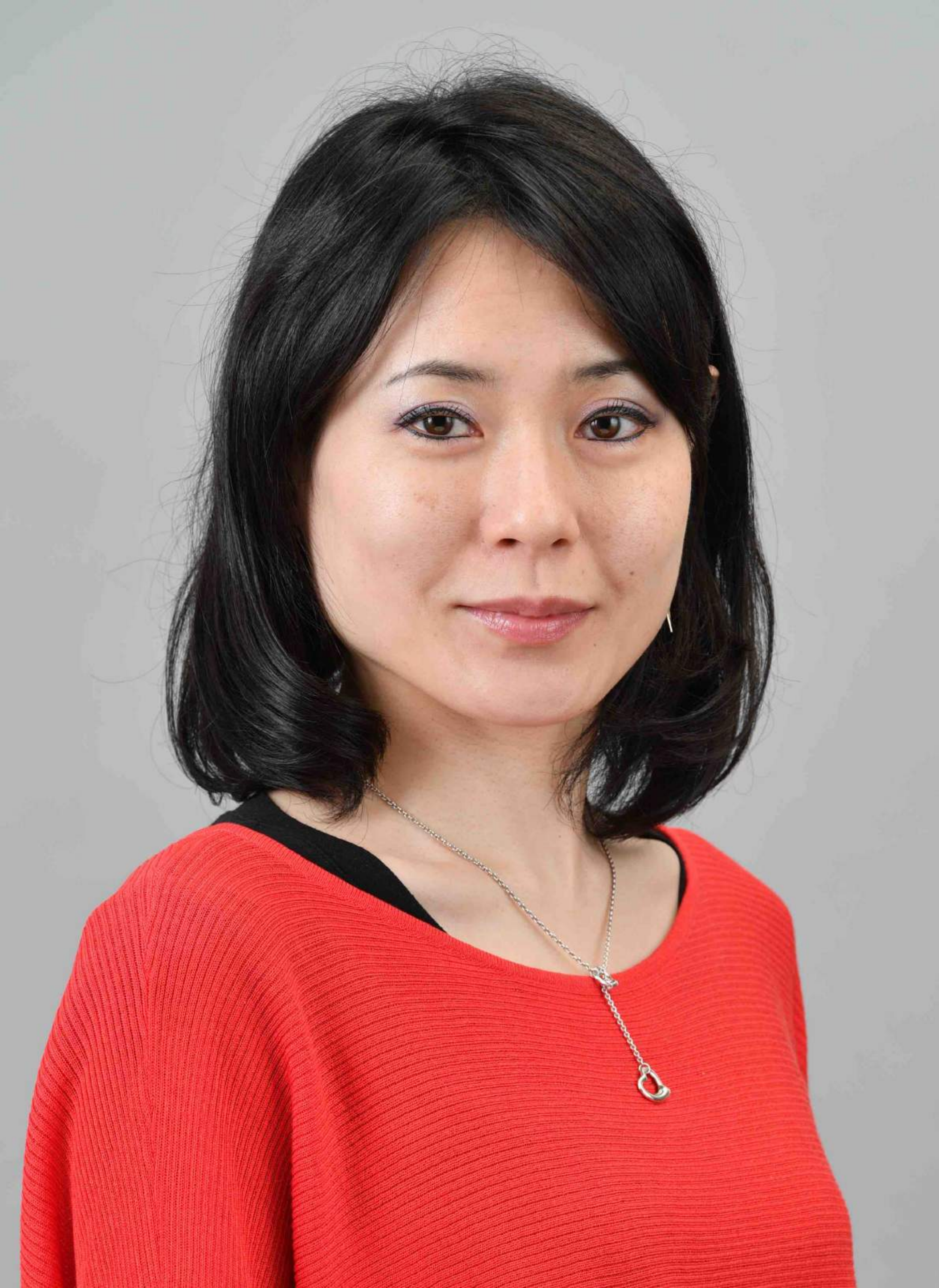}}]{Megumi Kaneko}
received her B.S. and MSc. degrees in communication engineering in 2003 and 2004 from Institut National des T\' {e}l\'{e}communications (INT, T\'{e}l\'{e}com SudParis), France, jointly with a MSc. from Aalborg University, Denmark, where she received her Ph.D. degree in 2007. In May 2017, she obtained her HDR degree (French Doctoral Habilitation for Directing Researches at Professor position) from Paris-Sud University, France. From January to July 2007, she was a visiting researcher in Kyoto University, Kyoto, Japan, and a JSPS post-doctoral fellow from April 2008 to August 2010. From September 2010 to March 2016, she was an Assistant Professor in the Department of Systems Science, Graduate School of Informatics, Kyoto University. She is currently an Associate Professor at the National Institute of Informatics as well as the Graduate University of Advanced Studies (Sokendai), Tokyo, Japan. Her research interests include wireless communications, radio resource and interference management, wireless sensing and cross-layer network protocols. She received the 2009 Ericsson Young Scientist Award, the IEEE Globecom 2009 Best Paper Award, the 2011 Funai Young Researcher's Award, the WPMC 2011 Best Paper Award, the 2012 Telecom System Technology Award and the 2016 Inamori Foundation Research Grant.
\end{IEEEbiography}

\begin{IEEEbiography}[{\includegraphics[width=1in,height=1.25in,clip,keepaspectratio]{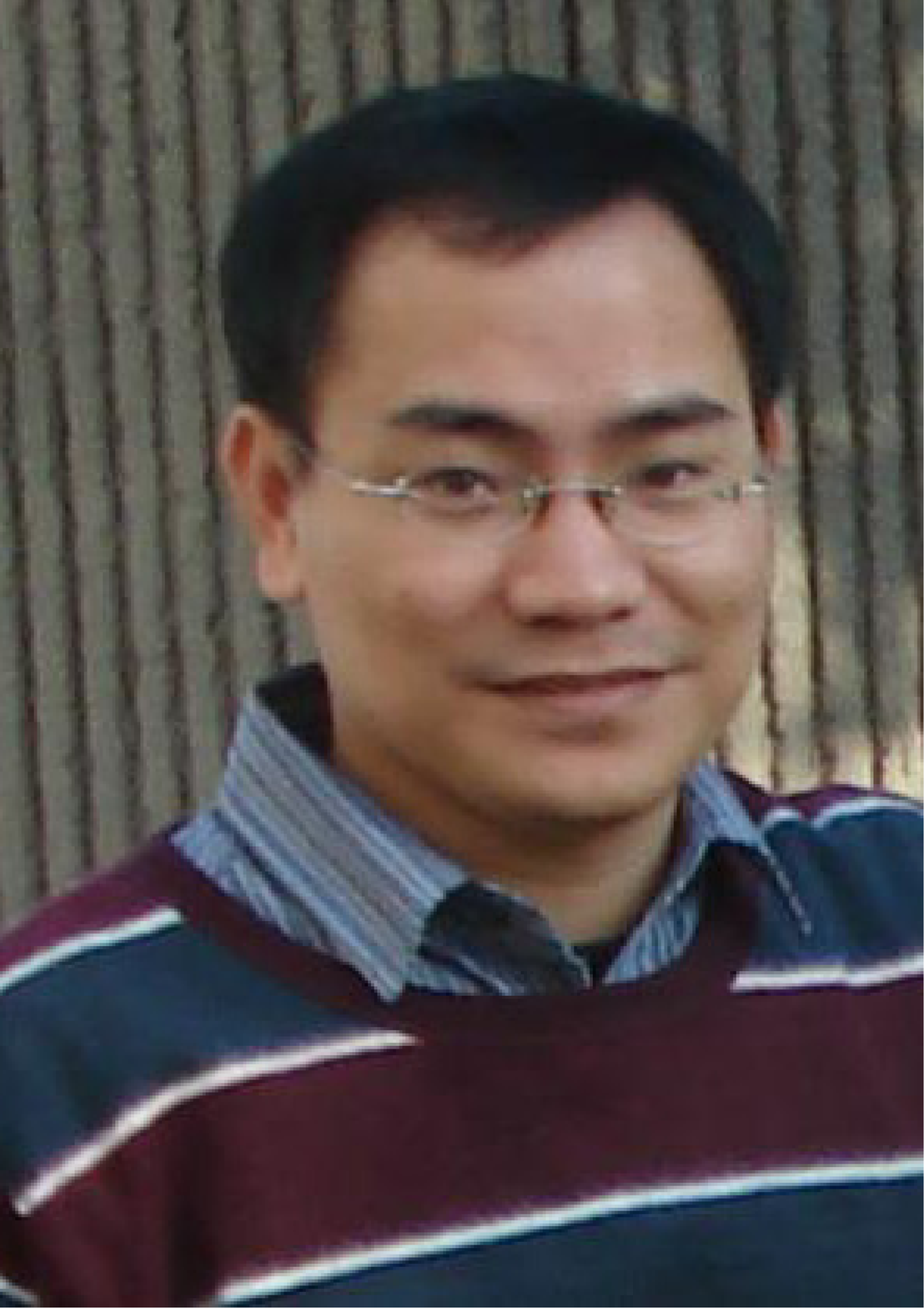}}]{Xu Chen}
received the Ph.D. degree in information engineering from The Chinese University of Hong Kong in 2012. He was a Post-Doctoral Research Associate with Arizona State University, Tempe, USA, from 2012 to 2014, and a Humboldt Fellow with the University of G\"ottingen, Germany, from 2014 to 2016. He is currently a Professor with the School of Data and Computer Science, Sun Yat-sen University, Guangzhou, China. He received 2017 IEEE ComSoc Pacifc-Asia Outstanding Young Researcher Award, 2017 IEEE ComSoc Young Professional Best Paper Award, 2014 Hong Kong Young Scientist Runner-Up Award, Best Paper Award in IEEE ICC 2017, Best Paper Runner-Up Award in IEEE INFOCOM 2014, and Honorable Mention Award in IEEE ISI 2010.
\end{IEEEbiography}
\begin{IEEEbiography}[{\includegraphics[width=1in,height=1.25in,clip,keepaspectratio]{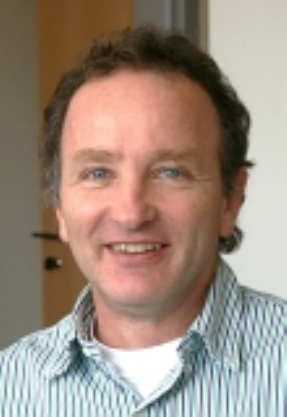}}]{Dieter Hogrefe}
received the Ph.D. degree in computer science and mathematics from the University of Hannover, Germany, in 1985. From 1983 to 1986, he was with the SIEMENS Research Center, Munich, and involved in the area of analysis of telecommunication systems. He was responsible for the protocol simulation and analysis of the CCS No. 7. From 1996 to 2010, he was the Chairman of the Technical Committee Methods for Testing and Specication with the European telecommunication Standards Institute, ETSI. From 2011 to 2013, he was the Dean of the Faculty of Mathematics and Computer Science. He held a full professor position with the University of Bern, the University of L\"ubeck, and the University of G\"ottingen and visiting positions with the University of Dortmund, Technical University Budapest, UC Berkeley, and Hamilton University. He has been a Full Professor (C4) for Telematics with the University of G\"ottingen since 2002. Since 2003, he has been the Director of the Institute of Computer Science. He has published numerous papers and two books on Internet technology, security of wireless sensor networks, analysis, simulation and testing of formally specified communication systems. His research activities are directed towards computer networks and communication software engineering.\end{IEEEbiography}
\end{document}